\documentclass[11pt,reqno]{article}
\usepackage{bm, booktabs}
\usepackage{amsthm}
\usepackage{amsmath}
\usepackage{amssymb}
\usepackage[round]{natbib}
\usepackage{color}
\usepackage{authblk}
\usepackage{microtype}
\usepackage{graphicx}
\usepackage{float}
\usepackage{prodint}
\usepackage{epstopdf}
\usepackage{lscape} 
\usepackage{tabularx}

\usepackage{geometry}
\newcommand{\rnc}{\renewcommand}
\newcommand{\nc}{\newcommand}
\newcommand{\mrm}{\mathrm}

\nc{\mb}{\mathbb}
\nc{\mc}{\mathcal}
\nc{\E}{\mb{E}}
\nc{\N}{\mb{N}}
\nc{\R}{\mb{R}}
\nc{\Q}{\mb{Q}}
\rnc{\P}{\mrm P}
\rnc{\d}{\mrm d}
\nc{\C}{\mc{C}}
\nc{\D}{\mc{D}}
\nc{\B}{\mc{B}}
\nc{\vbeta}{\bm \beta}
\nc{\vtheta}{\bm \theta}
\nc{\vX}{\bm X}
\nc{\vy}{\bm y}
\nc{\vU}{\bm U}
\nc{\vI}{\bm I}
\nc{\vE}{\bm E}
\nc{\ve}{\bm e}
\nc{\vV}{\bm V}
\nc{\vv}{\bm v}
\nc{\vS}{\bm S}
\nc{\vSigma}{\bm \Sigma}
\nc{\oPo}{\stackrel{\mrm p}{\rightarrow}}
\nc{\oWo}{\stackrel{w}{\rightarrow}}
\nc{\oDo}{\stackrel{d}{\longrightarrow}}
\nc{\eff}{\|F\|}
\definecolor{darkgreen}{rgb}{0.2, 0.5, 0.33}
\nc{\dd}{\color{darkgreen}}

\nc{\mbf}{\boldsymbol}
\nc{\trans}{'}

\def\E{{ E }}
\def\R{{ \mathbb{R} }}
\def\N{{ \mathbb{N} }}

\def\P{ P }

\def\E{ E }

\def\mbf{\boldsymbol}

\newtheorem{theorem}{Theorem}
\newtheorem{corollary}{Corollary}
\newtheorem{prop}{Proposition}

\newcommand\blfootnote[1]{%
  \begingroup
  \renewcommand\thefootnote{}\footnote{#1}%
  \addtocounter{footnote}{-1}%
  \endgroup
} 

\newcolumntype{H}{>{\setbox0=\hbox\bgroup}c<{\egroup}@{}}  
\newcolumntype{Z}{>{\setbox0=\hbox\bgroup}c<{\egroup}@{\hspace*{-\tabcolsep}}}

\begin{document}

\title{\Large \bf Studentized Permutation Method for Comparing Restricted Mean Survival Times with Small Sample from Randomized Trials 
 }
\author[1,$*$]{Marc Ditzhaus}
\author[2]{Menggang Yu}
\author[3]{Jin Xu}

\affil[1]{Department of Statistics, TU Dortmund University, Germany.}
\affil[2]{Department of Biostatistics and Medical Informatics, University of Wisconsin, Madison WI, USA}
\affil[3]{Key Laboratory of Advanced Theory and Application in Statistics and Data Science - MOE and School of Statistics, East China Normal University, China}
\maketitle

\begin{abstract}
	Recent observations, especially in cancer immunotherapy clinical trials with time-to-event outcomes, show that the commonly used proportional hazard assumption is often not justifiable, hampering an appropriate analyse of the data by hazard ratios. An attractive alternative advocated is given by the restricted mean survival time (RMST), which does not rely on any model assumption and can always be interpreted intuitively. As pointed out recently by \cite{horiguchi2020permutation}, methods for the RMST based on asymptotic theory suffer from inflated type-I error under small sample sizes. To overcome this problem, they suggested a permutation strategy leading to more convincing results in simulations. However, their proposal requires an exchangeable data set-up between comparison groups which may be limiting in practice. In addition, it is not possible to invert their testing procedure to obtain valid confidence intervals, which can provide more in-depth information. In this paper, we address these limitations by proposing a studentized permutation test as well as the corresponding permutation-based confidence intervals. In our extensive simulation study, we demonstrate the advantage of our new method, especially in situations with relative small sample sizes and unbalanced groups. Finally we illustrate the application of the proposed method by re-analysing data from a recent lung cancer clinical trial.
\blfootnote{${}^*$ e-mail: marc.ditzhaus@tu-dortmund.de}
\end{abstract}

\noindent{\bf Keywords:} hazard ratio, permutation methods, restricted mean survival time, survival analysis, time-to-event outcomes.


\section{Introduction}

While the log-rank test and hazard ratios were the gold standard in time-to-event analysis for a long time, there is a recent trend towards alternative methods not relying on the proportional hazard assumption. The reason for this change are recently observed violations of the proportional hazard assumption in real data. For example, \cite{trinquart2016comparison} analysed 54 phase III oncology clinical trials from five leading journals and in 13 ($24\%$) of them the proportional hazard assumption could be rejected significantly. Especially in immunotherapy trials, a delayed treatment effect often lead to a violation of the proportional hazard assumption \citep{mick2015statistical,alexander2018hazards} and suchlike could also be observed when comparing bone marrow transplant and chemotherapy for hematologic malignancies \citep{zittoun1995autologous,scott2017myeloablative}. More classical and known effect sizes as landmark survival \citep{taori2009landmark} and the median survival time \citep{brookmeyer1982confidence,chen2016comparing,ditzhaus2020inferring} provide rather a snapshot for a time point than information about the complete Kaplan--Meier curves. 

This may be one of the reasons why the restricted mean survival time (RMST)
\citep{royston2011use,royston2013restricted,uno2014moving,a2016restricted,ZhaoETAL2016,kim2017restricted}, the integral of the
Kaplan--Meier-curve over a clinically relevant time window, gets more and more attention lately. Since this summary measure is ``arguably more helpful for clinical decision-making and more easily understood by patients'' \citep{stensrud2020test}, methods based on it ``should
be routinely reported in randomized trials with time-to-event outcomes'' \citep{trinquart2016comparison}.

Recently, \cite{horiguchi2020permutation} pointed out that ``there is a notable inflation of the type-I error rate'' under small sample sizes when the asymptotic methods are used for two-sample comparisons in terms of the RMST. To overcome this problem, they suggested a permutation approach showing a significant improvement regarding the type-I error control in simulations. For small sample sizes, permutation methods are popular tools since they guarantee exact testing procedures \citep{lehmann2006testing,hemerik2018} for exchangeable data. However, the assumption of exchangeability reduces the applicability notably, especially in the context of survival data due to the following reasons: (i) the censoring distributions may differ between comparison groups, e.g. due to side effects or other dropout reasons; (ii) the survival distributions may differ even when the RMSTs of the comparison groups may coincide under the null; and (iii) confidence intervals for the quantity of interest, here the difference of the RMSTs, cannot be derived as \cite{horiguchi2020permutation} mentioned: ``Further research to develop methods for constructing confidence intervals for RMST
difference with a small sample data is warranted''.

In this paper, we propose a studentized permutation method that does not require the exchangeability assumption, thus overcoming these limitations. Our method builds on existing work for two-sample comparisons \citep{neuhaus:1993,janssen:1997studentized,chungRomano2013,dopa2018,ditzhaus2020more,ditzhaus2020bootstrap}, which was recently extended to one-way layouts \citep{chungRomano2013,pauly2020asymptotic} and even to general factorial designs \citep{paulyETAL2015,smaga2017diagonal,berrett2020conditional,ditzhaus2019qanova}. 
While the finite exactness of permutation tests is preserved, studentized permutation tests are often shown to be still asymptotically valid for non-exchangeable data. Moreover, they exhibit a satisfactory performance under small sample sizes in simulations. The aim of this paper is to justify theoretically and in an extensive simulation study that the concept of studentization can also be used to extend the test of \citet{horiguchi2020permutation} to non-exchangeable data settings and, furthermore, to derive confidence intervals. For that purpose, we employ the empirical process theory \citep{vaartWellner1996} which can easily handle even tied data, e.g. survival times rounded to days, months etc.

The paper is organized as follows. First, our methodology is presented in Section~\ref{sec:setup}. Therein, we explain in detail why the permutation test of \cite{horiguchi2020permutation} may fail for general non-exchangeable data and how studentization solves for that. Moreover, permutation-based confidence intervals for the difference and the ratio of RMSTs are presented. To empirically assess the performances of the proposed test and confidence interval, we conducted an extensive simulation study comparing the asymptotic and the two permutation methods in Section~\ref{sec:sim}. Their applications are illustrated by analyzing data from a recent lung cancer trial in Section~\ref{sec:data_example}. Finally, we give some final remarks and discuss possible future extensions in Section~\ref{sec:discussion}. All proofs and some additional simulation results are given in the supplement.

\section{Methodology}
\label{sec:setup}
We consider the two-sample survival set-up given by mutually independent survival and censoring times
\begin{align*}
T_{ij}\sim S_i,\quad C_{ij}\sim G_{i}, \quad i=1,2; \ \  j=1,\ldots,n_i,
\end{align*}
respectively. Here, $S_i$ and $G_i$ denote the survival functions for the survival and censoring times of the $i$th group, respectively. Both are not necessarily continuous and ties in the data are explicitly allowed, e.g. survival times rounded to days, months etc. Based on the right-censored event times $X_{ij}=\min(T_{ij},C_{ij})$ and the censoring statuses $\delta_{ij}=\mathbf{1}\{X_{ij}=T_{ij}\}$, we would like to infer differences between the two groups in terms of their RMSTs
\begin{align*}
\mu_i= \int_0^\tau S_i(t) \,\mathrm{ d }t \quad (i=1,2)
\end{align*}
over a pre-specified time window $[0,\tau]$, which is practically relevant (e.g. $\tau=2$ years). Thereby, it needs to be guaranteed that the event times $X_{ij}$ larger than $\tau$ are observable with a positive probability $P(X_{ij}\geq \tau)>0$. In practice, a typical choice for $\tau$ is the end-of-study time. While $\tau$ is usually be chosen as a pre-specified constant allowing a straight-forward interpretation of $\mu_i$, \cite{tian2020empirical} discuss an empirical choice of $\tau$, e.g. the largest observed time, under appropriate regularity assumptions on the censoring distribution.

The RMST can be naturally estimated by plugging-in the Kaplan-Meier estimator $\widehat S_i$:
\begin{align*}
\widehat \mu_i = \int_0^\tau \widehat S_i(t) \,\mathrm{ d } t \quad (i=1,2).
\end{align*}
Asymptotic inference for this estimator relies on a normal approximation, which can be justified by martingale arguments \citep{abgk} combined with the continuous mapping theorem. In fact, under the assumption of non-vanishing groups, i.e. $n_i/n\to \kappa_i\in(0,1)$ as $n\to\infty$, which is supposed throughout the paper, we obtain
\begin{equation}\label{eqn:asym_diff}
\sqrt{n} \left\{ (\widehat \mu_1 - \widehat \mu_2) - (\mu_1 - \mu_2) \right\} \stackrel{d}{\rightarrow} Z \sim N(0,\sigma^2), \quad \sigma^2 = \sigma_1^2 + \sigma_2^2.
\end{equation}
Here, $\sigma_i^2$ denotes the asymptotic variance of $\sqrt{n}(\widehat \mu_i - \mu_i)$ and is given by
\begin{align*}
\sigma^2_i = \kappa_i^{-1}\int_0^\tau \left\{\int_x^\tau S_i(t) \,\mathrm{ d }t\right\}^2 \frac{1}{\{1-\Delta A_i(x)\}G_{i-}(x)S_{i-}(x)} \;\mathrm{ d }A_i(x) \quad (i=1,2),\nonumber
\end{align*}
where $A_i=-\log(S_i)$ is the cumulative hazard rate function and $\Delta A_i(x) = A_i(x) - A_{i-}(x)$ is its increment in $t$. Moreover, $G_{i-}$, $S_{i-}$ and $A_{i-}$ denote the left-continuous versions of $G_i$, $S_i$ and $A_i$, respectively, e.g., $G_{i-}(t)=\P(C_{i1}\geq t)$ (c.f. $G_i(t)=\P(C_{i1}>t)$).

While the convergence in \eqref{eqn:asym_diff}  is well established \citep[see e.g.][]{ZhaoETAL2016} for continuously distributed survival and censoring times, it even remains true when ties are allowed. See the supplement for a detailed proof. The variance can be estimated straightforwardly by replacing $S_i$, $G_i$ and $A_i$ by their respective Kaplan--Meier ($\widehat S_i$, $\widehat G_i$) and Nelson--Aalen ($\widehat A_i$) estimators. In detail, $\widehat \sigma = \widehat \sigma_1^2 + \widehat \sigma_2^2$ and
\begin{equation}\label{eqn:sigma-i-hat}
\widehat\sigma^2_i =  \frac{n}{n_i} \int_0^\tau \left\{\int_x^\tau \widehat S_i(t) \,\mathrm{ d }t\right\}^2 \frac{1}{\{1-\Delta \widehat A_i(x)\}\widehat S_{i-}(x)\widehat G_{i-}(x)} \;\mathrm{ d }\widehat A_i(x).
\end{equation}
Combining \eqref{eqn:asym_diff} and \eqref{eqn:sigma-i-hat}, we obtain an asymptotically valid test $\varphi = \mathbf{1}\{ \sqrt{n}|\widehat\mu_1 - \widehat \mu_2| / \widehat\sigma > z_{1-\alpha/2}\}$ for the null hypothesis of equal RMSTs:
\begin{align*}
\mathcal H_0: \mu_1=\mu_2.
\end{align*}
Here, $z_{1-\alpha/2}$ denotes the $(1-\alpha/2)$-quantile of a standard normal distribution.  However, for small sample sizes this test has an inflated type-I error control, as seen in \citet{horiguchi2020permutation} and in Section~\ref{sec:sim}. To tackle this problem, \cite{horiguchi2020permutation} proposed a permutation approach. In the next subsection, we discuss their permutation approach as well as its limitations and propose an improved permutation strategy, both hypothesis testing and confidence interval construction.

\subsection{Unstudentized permutation test and its  studentized version}\label{sec:permutation}
Following the idea of exact permutation tests \citep{lehmann2006testing,hemerik2018}, \cite{horiguchi2020permutation} recently proposed a permutation test for $\mathcal H_0:\mu_1=\mu_2$, which we call the \textit{unstudentized} test hereafter.

In detail, 
given the observed data $({\bm X}, {\bm \delta})\equiv \big \{(X_{ij},\delta_{ij}):\ i=1, 2;\,  j=1, \dots, n_i \big \},$ 
let $({\bm X}^\pi, {\bm \delta}^\pi)\equiv \big \{(X_{ij}^\pi,\delta_{ij}^\pi):\ i=1, 2;\,  j=1, \dots, n_i \big \}$  be its permutated version corresponding to a scramble of the treatment indicator. Note that the permutation is at the subject level and $(X_{ij},\delta_{ij})$ are permutated in pairs. 
\cite{horiguchi2020permutation} suggested using the permutation test $\varphi^\pi_{\text{HU}}=\mathbf{1}\{|\widehat \mu_1 - \widehat\mu_2|> q^\pi_{1-\alpha,HU}\}$ in case of small sample sizes, where $q^\pi_{1-\alpha,HU}$ is the $(1-\alpha)$-quantile of the permutation distribution $t\mapsto P\{|\widehat \mu_1^\pi - \widehat\mu_2^\pi|\leq t \vert ({\bm X}, {\bm \delta})\}$ given the observed data $({\bm X}, {\bm \delta})$. Here, $\widehat \mu_i^\pi$, $\widehat S_i^\pi$ denote the permutation counterparts of the original estimators by replacing the data $({\bm X}, {\bm \delta})$ with a permuted sample $({\bm X}^\pi, {\bm \delta}^\pi)$. 

Such permutation tests are known to be finitely exact, i.e. the type-I error is controlled not only asymptotically but for every fixed sample size, under exchangeable data. In the context of right-censored survival data, exchangeability implies equal survival and censoring distributions between the groups, respectively, i.e. $S_1=S_2$ and $G_1=G_2$. This is obviously a much stronger assumption on both the interested time-to-event outcome and the censoring distribubtions. In our context of RMST comparison, having potentially crossing survival curves in mind, it may occur that the null hypothesis $\mathcal H_0:\mu_1=\mu_2$ is true despite $S_1\neq S_2$ holds, as shown by the four examples in Figure~\ref{fig:survival_times}. In addition, the assumption of equal censoring distributions alone is also too restrictive, since side effects related to the treatment may lead to different drop-out rates for example.  
An additional disadvantage is that this unstudentized permutation strategy cannot be used to obtain valid confidence intervals because the fact $\mu_1\neq\mu_2$ clearly violates the exchangeability assumption.

\begin{figure}
	\centering
	\begin{minipage}{0.47\textwidth}
		\includegraphics[width=\textwidth]{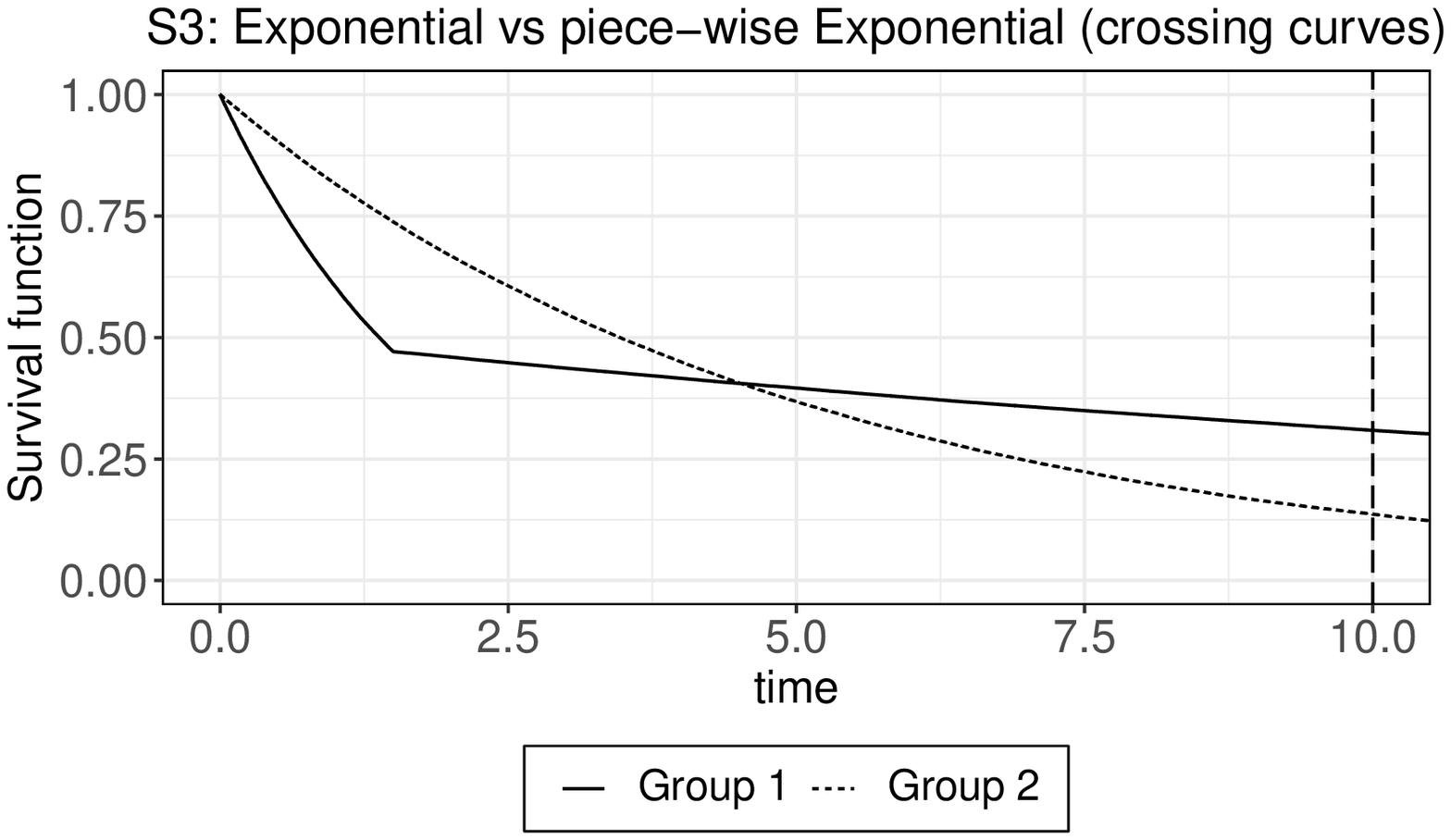}
	\end{minipage}
	\begin{minipage}{0.47\textwidth}
		\includegraphics[width=\textwidth]{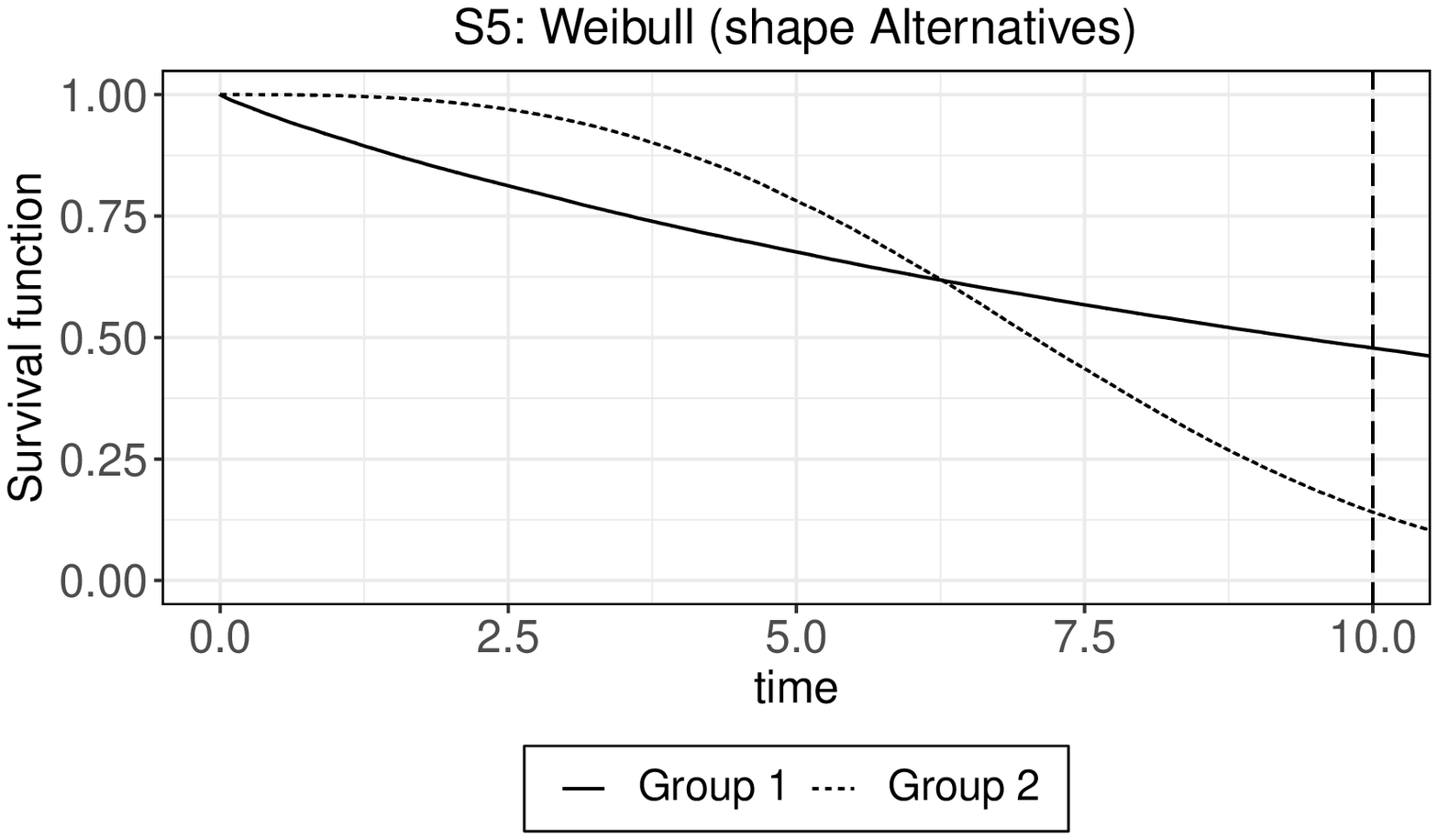}
	\end{minipage}
	
	\begin{minipage}{0.47\textwidth}
		\includegraphics[width=\textwidth]{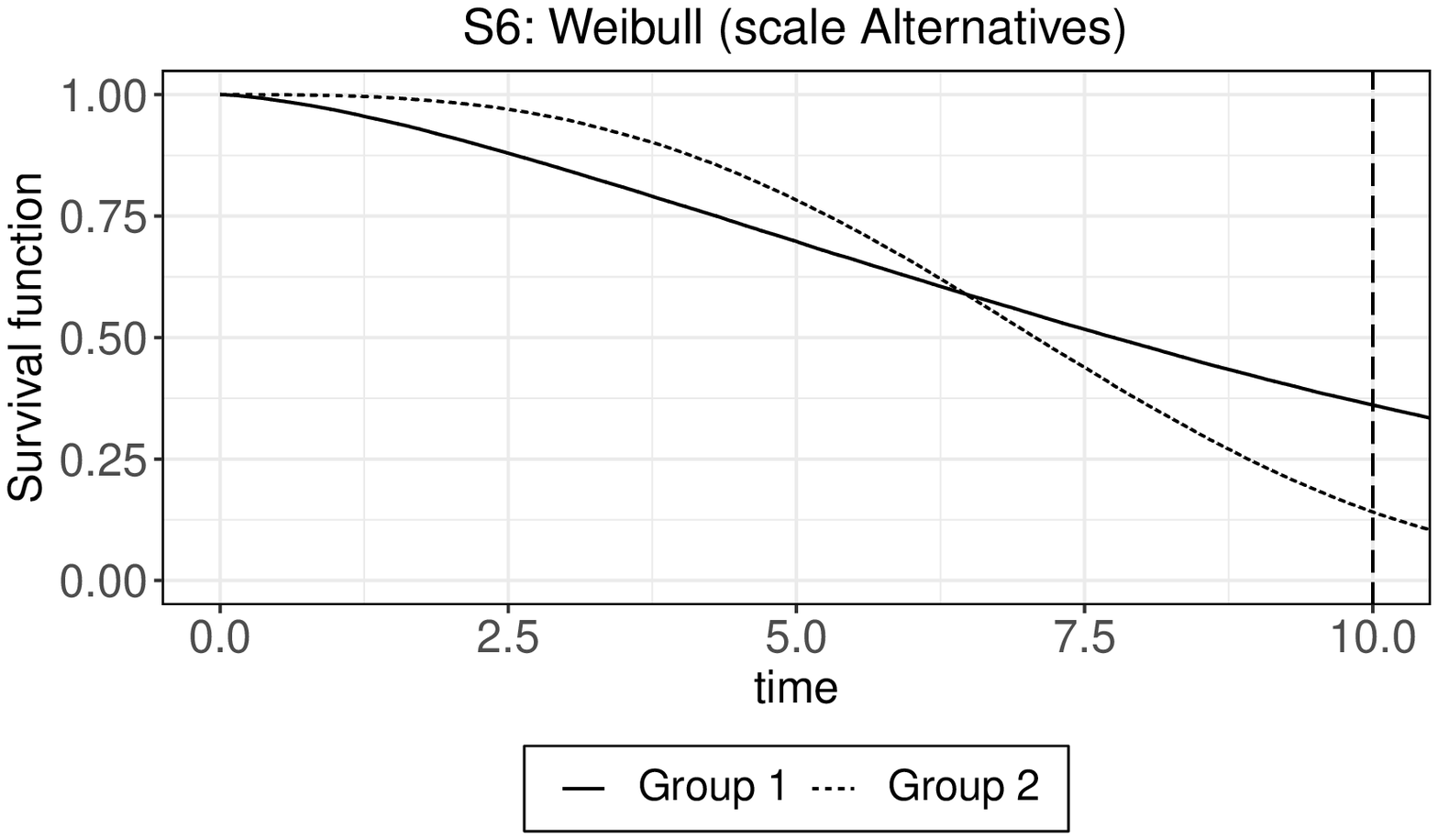}
	\end{minipage}
	\begin{minipage}{0.47\textwidth}
		\includegraphics[width=\textwidth]{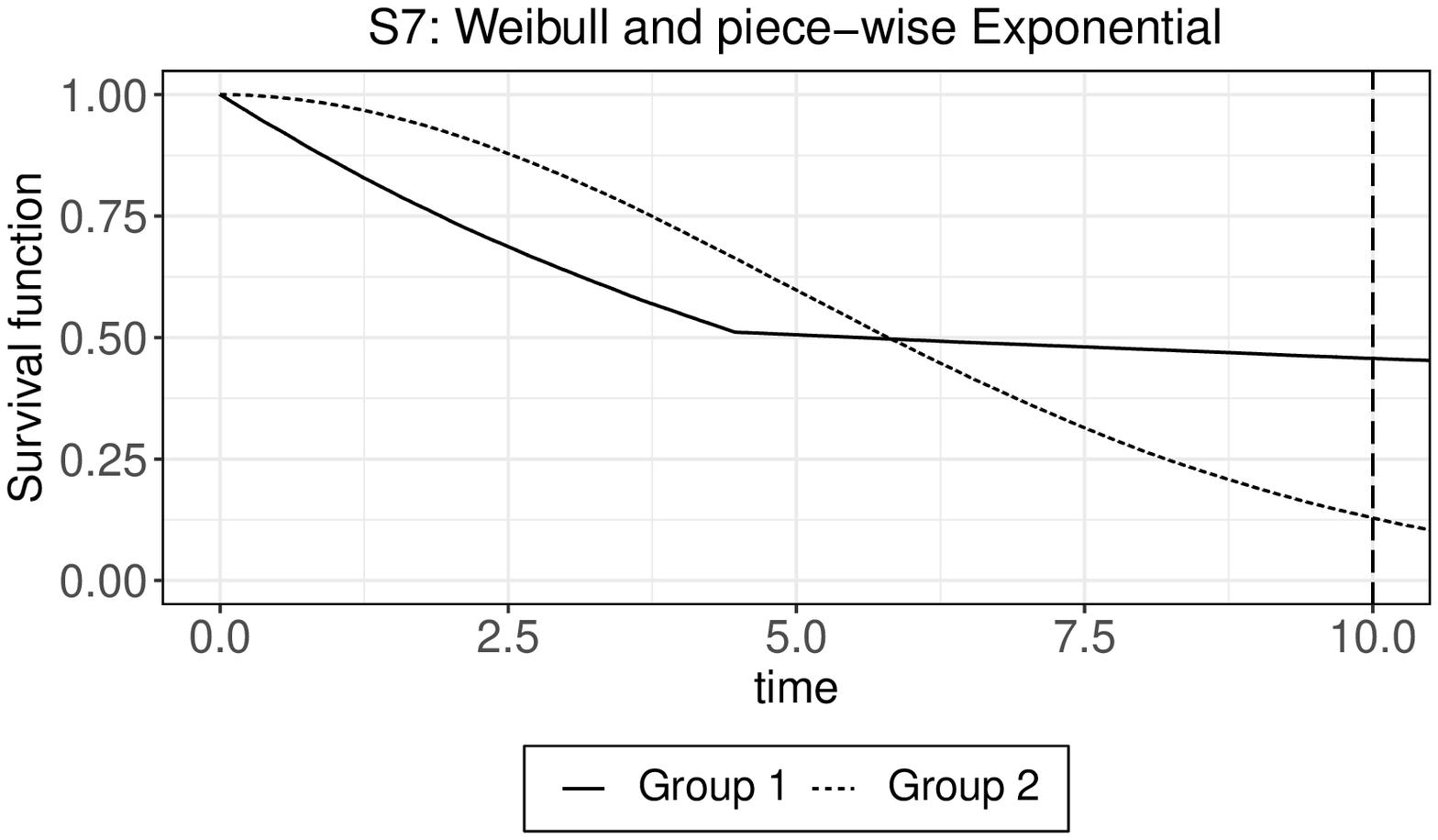}
	\end{minipage}
	\caption{Four examples, for which the groups' survival curves are different but their restricted mean survival time over $[0,10]$ coincides. The examples correspond to Scenarios S3, S5, S6 and S7 from the simulation study, see Section~\ref{sec:sim}}\label{fig:survival_times}
\end{figure}

To address all these issues, we propose a \textit{studentized} permutation test. To explain our idea, we need to understand first the asymptotic behavior of the permutated, unstudentized statistic, here $\widehat \mu_1^\pi - \widehat\mu_2^\pi$, under non-exchangeable settings. For that purpose, we introduce the pooled Kaplan--Meier estimator $\widehat S$ and the pooled Nelson--Aalen-estimator $\widehat A$. 
In detail, let $N(t)=\sum_{i,j} \delta_{ij}\mathbf{1}\{X_{ij}\leq t\}$ be the number of events up until $t$ and $Y(t)=\sum_{i,j}\mathbf{1}\{X_{i,j}\geq t\}$ be the number of individuals under risk at time $t$. Moreover, let $t_1,\ldots,t_d$, $d\in\N$, be the distinctive time points within $\mathbf{X}$. Then $\widehat S(t)=\prod_{k:t_k\leq t}[1- \Delta N(t_k)/Y(t_k)]$ and $ \widehat A(t)=\sum_{k:t_k\leq t} \Delta N(t_k)/Y_i(t_k)$. Now, define $y(t) = \sum_{i=1}^2 \kappa_i S_{i-}(t) G_{i-}(t)$ and $\nu(t) = \sum_{i=1}^2 \kappa_i \int_0^t G_{i-}(s) \,\mathrm{ d }F_i(s)$, where $F_i=1-S_i$. Combining the Glivenko-Cantelli Theorem and the continuous mapping theorem we obtain almost surely that $\widehat S(t)$ and $\widehat A(t)$ converge uniformly on $[0,\tau]$ to
$S(t) = \exp[ - A(t)]$ and $A(t) = \int_0^t 1/y(s)\mathrm{d}\nu(s)$, respectively, see the supplement for more details.
%
%
Having these additional notations at hand, we are now able to derive the asymptotic limit of the permuted, unstudentized statistic $\widehat \mu_1^\pi - \widehat\mu_2^\pi$:

\begin{theorem}\label{theo:perm_unstud}
	Under $\mathcal H_0:\mu_1=\mu_2$ as well as under $\mathcal H_1:\mu_1\neq \mu_2$, we have the following conditional convergence in distribution
	\begin{align*}
	\sqrt{n}(\widehat\mu_1^\pi - \widehat \mu_2^\pi)  \overset{ d}{\rightarrow} Z_\text{perm} \sim N(0,\sigma^2_\text{perm}), \textrm{ as }n\to\infty,
	\end{align*}
	given the data in probability, where the limiting variance is given by
	\begin{align*}
	\sigma_\text{perm}^2 = \frac{1}{\kappa_1\kappa_2} \int_0^\tau \left\{\int_x^\tau S(t) \,\mathrm{ d }t\right\}^2 \frac{1}{\{1-\Delta A(x)\}y(t)} \;\mathrm{ d } A(x).
	\end{align*}
\end{theorem}
In the special case $S_1=S_2$ and $G_1=G_2$, the variances $\sigma^2$ in (\ref{eqn:asym_diff}) and $\sigma_\text{perm}^2$ coincide. But, in general, they are different. Thus, applying the unstudentized permutation test for a non-exchangeable setting may lead to a systematic error, which is caused by a different variance of the permuted statistic. However, this can be solved by studentization, i.e. by including an appropriate variance estimator in the original test statistic as well as in its permutation counterpart. In fact, it can be shown that the permutation counterpart $\widehat \sigma^{\pi 2}$ of the variance estimator $\widehat \sigma^2$ converges, given the observed data, to the variance $\sigma^2_\text{perm}$ from Theorem \ref{theo:perm_unstud}. See the supplement for a detailed proof. In other words, inclusion of the variance estimator in the permutation step corrects the wrong variance. Consequently, we obtain
\begin{theorem}\label{theo:perm}
	Under $\mathcal H_0:\mu_1=\mu_2$ as well as under $\mathcal H_1:\mu_1\neq \mu_2$ we have the following conditional convergence in distribution
	\begin{align*}
	\sqrt{n}(\widehat\mu_1^\pi - \widehat \mu_2^\pi)/\widehat\sigma^\pi  \overset{ d}{\rightarrow} Z_\text{perm} \sim N(0,1) \textrm{ as }n\to\infty,
	\end{align*}
	given the observed data in probability. 
\end{theorem}

From Theorem~\ref{theo:perm} we obtain that $\sqrt{n}{|\widehat\mu_1 - \widehat \mu_2-(\mu_1-\mu_2)|}/{ \widehat\sigma}$ and $\sqrt{n}|\widehat\mu_1^\pi - \widehat \mu_2^\pi|/\widehat\sigma^\pi$ have the same asymptotic distribution, namely $|Z|$ for $Z\sim N(0,1)$. Moreover, we like to point out that this is true under alternatives as well, which allows us to formulate even asymptotically valid confidence intervals. For that purpose, let $q^\pi_{1-\alpha}$  denote the $(1-\alpha)$-quantile of the conditional distribution $t\mapsto  \P\{ \sqrt{n}|\widehat\mu_1^\pi - \widehat \mu_2^\pi| / \widehat\sigma^\pi \leq t  \vert \  ({\bm X}, {\bm \delta})\}$. Then the studentized permutation test $\varphi^\pi$ and the permutation-based confidence interval $I^\pi$ for $\mu_1-\mu_2$ are given by
\begin{align*}
\varphi^\pi = \mathbf{1}\Bigl\{ \sqrt{n}\frac{|\widehat\mu_1 - \widehat \mu_2|}{ \widehat\sigma} > q^\pi_{1-\alpha} \Bigr\}, \quad I^\pi = \Bigl[ \widehat \mu_1 - \widehat \mu_2 \pm n^{-1/2}\widehat\sigma\: q^\pi_{1-\alpha}   \Bigr].
\end{align*}

Combining \eqref{eqn:asym_diff},  Theorem~\ref{theo:perm}, as well as Lemma 1 and Theorem 7 of \cite{janssenPauls2003}, we can deduce that the conditional quantile $q^\pi_{1-\alpha}$ tends to $z_{1-\alpha/2}$ and we obtain:

\begin{corollary}\label{cor:tests+confidence_intervals}
	(i) The permutation test $\varphi^\pi$ has asymptotic level $\alpha$ for $\mathcal{H}_0: \mu_1= \mu_2$ and is consistent for general alternatives  $\mathcal{H}_1: \mu_1\neq \mu_2$, i.e. $
	E_{\mathcal{H}_0}(\varphi^\pi) \to \alpha$ and $ E_{\mathcal{H}_1}(\varphi^\pi) \to 1$ as $n\to\infty$.	
	(ii) The permutation-based confidence interval $I^\pi$ has asymptotic confidence level $1-\alpha$, i.e.,
	$P(\mu_1-\mu_2\in I^\pi)\to 1-\alpha$ as $n\to\infty$.
\end{corollary}

\subsection{Permutation-based confidence intervals for RMST ratio}\label{sec:ratio}

In this subsection, we briefly explain how the permutation strategy can also be adopted to obtain confidence intervals for the ratio  $\mu_1/\mu_2$. While the studentization idea directly applied to the ratio would lead to inappropriate confidence intervals for a ratio, i.e. $\widehat\mu_1 / \widehat \mu_2 \pm D_n$, we consider the log-transformation $\log(\widehat\mu_1) - \log(\widehat \mu_2)$ instead. Analogous to \eqref{eqn:asym_diff}, it can be shown that
\begin{equation}\label{eqn:asym_ratio}
\sqrt{n}\left[ \left\{\log(\widehat \mu_1) - \log(\widehat \mu_2)\right\} - \left\{\log(\mu_1) - \log(\mu_2)\right\} \right] \to Z \sim N(0,\sigma_{\text{rat}}^2), \quad \sigma_{\text{rat}}^2 = \frac{\sigma_1^2}{\mu_1^2} + \frac{\sigma_2^2}{\mu_2^2}.
\end{equation}
The asymptotic variance can be estimated by $\widehat\sigma_{\text{rat}}^2 = (\widehat\sigma_1^2/\widehat\mu_1^2) + (\widehat\sigma_2^2/\widehat\mu_2^2)$. Consequently, an asymptotically valid confidence interval for $\mu_1/\mu_2$ and its studentized permutation counterpart are given respectively by
\begin{align*}
I_{\text{rat}} &= \Bigl[\exp\left\{\log(\widehat \mu_1) - \log(\widehat \mu_2) \pm n^{-1/2}\widehat\sigma_{\text{rat}} z_{1-\alpha/2} \right\}\Bigr], \\
I^\pi_{\text{rat}} &= \Bigl[\exp\left\{\log(\widehat \mu_1) - \log(\widehat \mu_2) \pm n^{-1/2}\widehat\sigma_{\text{rat}} q^\pi_{1-\alpha,\text{rat}}\right\}\Bigr],
\end{align*}
where $q^\pi_{1-\alpha}$ denotes the $(1-\alpha)$-quantile of the conditional distribution $t\mapsto  \P\{ \sqrt{n}|\log(\widehat\mu_1^\pi) - \log(\widehat \mu_2^\pi)| / \widehat\sigma^\pi_{\text{rat}} \leq t  \vert \  ({\bm X},{\bm \delta})\}$. Similarly to Corollary~\ref{cor:tests+confidence_intervals}, we can prove that the permutation-based confidence interval is asymptotically valid:
\begin{corollary}\label{cor:tests+confidence_intervals;rat}
	The permutation-based confidence interval $I^\pi_{\text{rat}}$ for the ratio $\mu_1/\mu_2$ has asymptotic confidence level $1-\alpha$, i.e., $P(\mu_1/\mu_2\in I^\pi_{\text{rat}})\to 1-\alpha$ as $n\to\infty$.
\end{corollary}

\section{Simulations}\label{sec:sim}
To complement our theoretical discussion from the previous section, we conducted an extensive simulation study to examine the performance of the permutation test as well as the permutation-based confidence intervals. For ease of presentation, we restricted ourselves to the difference of the RMSTs. Additional results for the ratio are deferred to the supplement.

\subsection{Setup}
We considered seven different choices for the survival times distribution:
\begin{enumerate}
	\item[S1] \textit{Exponential distributions (proportional hazards)}: $T_{11}\sim \text{Exp}(0.2)$ and $T_{21}\sim\text{Exp}(\lambda_{\delta,1})$.
	
	\item[S2] \textit{Exponential distribution vs piece-wise Exponential (late departures)}: $T_{11}\sim \text{Exp}(0.2)$ and $T_{21}$ has  piece-wise constant hazard function $\alpha_2(t) = 0.2\cdot \mathbf{1}\{t\leq 2\} + \lambda_{\delta,2}\: \mathbf{1}\{t>2\} $.
	
	\item[S3] \textit{Exponential distribution vs piece-wise Exponential (crossing curves)}: $T_{11}\sim \text{Exp}(0.2)$ and $T_{21}$ has  piece-wise constant hazard function $\alpha_2(t) = 0.5\cdot \mathbf{1}\{t\leq c_{\delta,1}\} + 0.05\cdot \mathbf{1}\{t>c_{\delta,1}\} $.
	
	\item[S4] \textit{Lognormal scale alternatives}: $T_{11}\sim \text{logN}(2,0.25)$ and $T_{11}\sim \text{logN}(\mu_\delta,0.25)$.
	
	\item[S5] \textit{Weibull shape alternatives (crossing curves)}: $T_{11}\sim \text{Weib}(3,8)$ and $T_{21}\sim \text{Weib}(\text{shape}_\delta,14)$.
	
	\item[S6] \textit{Weibull scale alternatives (crossing curves)}: $T_{11}\sim \text{Weib}(3,8)$ and $T_{21}\sim \text{Weib}(1.5,\text{scale}_\delta)$.
	
	\item[S7] \textit{Weibull vs piece-wise Exponential (crossing curves)}: $T_{11}\sim \text{Weib}(2,7)$ and $T_{21}$ has  piece-wise constant hazard function $\alpha_2(t) = 0.15\cdot \mathbf{1}\{t\leq c_{\delta,2}\} + 0.02 \cdot\mathbf{1}\{t>c_{\delta,2}\} $.
\end{enumerate}
The parameters $\lambda_{\delta,k}$, $c_{\delta,k}$, $\mu_\delta$, $\text{shape}_\delta$ and $\text{scale}_\delta$ depend on the difference $\delta=\mu_2-\mu_1$ of the RMSTs. For our simulations, we considered $\delta=0$ for the settings under the null hypotheses and $\delta\in\{0.5,1,1.5,2\}$ for the different alternative scenarios. See Figure~\ref{fig:survival_times} for an illustration of the Scenarios S3, S5, S6 and S7 with crossing curves under the null hypotheses ($\delta=0$). Under the null hypothesis ($\delta = 0$) Scenarios S1 and S2 coincide. That is why just one of the respective two scenarios was included in the simulation study whenever $\delta=0$ was considered. For the censoring, we chose the following three censoring configurations, see also Figure~\ref{fig:survival_censoring} for the respective survival functions: 
\begin{enumerate}
	\item[C1] \textit{unequally Weibull distributed censoring (Weib, uneq)}: $C_{11}\sim \text{Weib}(3,18)$ and $C_{21}\sim\text{Weib}(0.5,40)$.
	
	\item[C2] \textit{equally uniformly distributed censoring (Unif, eq)}: $C_{11}\sim \text{Unif}[0,25]$ and $C_{21}\sim\text{Unif}[0,25]$.
	
	\item[C3] \textit{equally Weibull distributed censoring (Weib, eq)}:  $C_{11}\sim \text{Weib}(3,15)$ and $C_{21}\sim\text{Weib}(3,15)$.
\end{enumerate}
For all simulations, we studied one balanced $\mathbf{n}_{\text{bal}} = (20,20)$ and two unbalanced, $\mathbf{n}_{\text{incr}} = (16,24)$ and $\mathbf{n}_{\text{decr}} = (24,16)$, sample size settings and, additionally, considered their multiples $K\mathbf{n}_{\text{bal}}, K\mathbf{n}_{\text{incr}}, K\mathbf{n}_{\text{decr}}$ with $K=2,4$ for larger sample sizes.

\begin{figure}
	\centering
	\begin{minipage}{0.33\textwidth}
		\includegraphics[width=\textwidth]{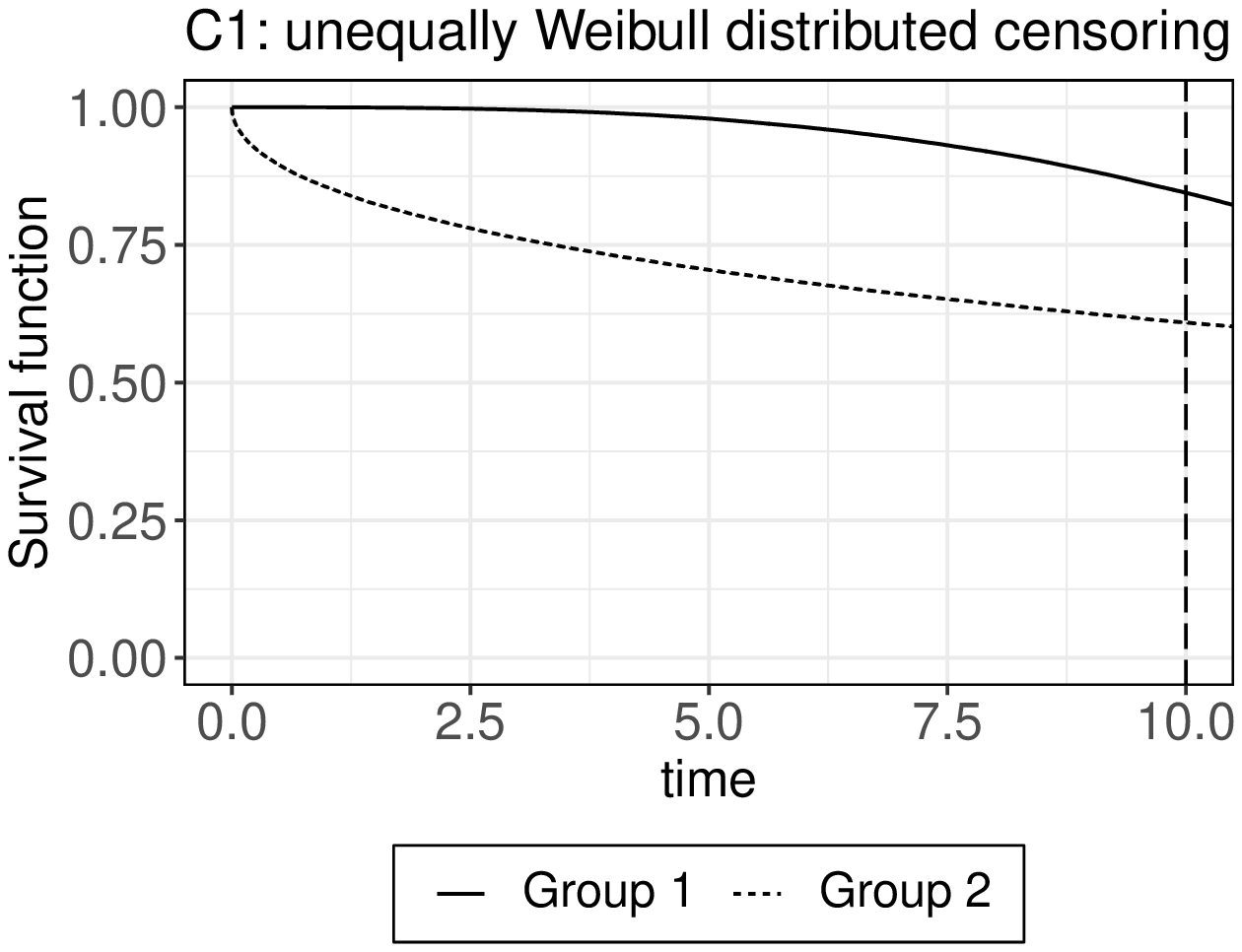}
	\end{minipage}
	\hspace{-0.3cm}
	\begin{minipage}{0.33\textwidth}
		\includegraphics[width=\textwidth]{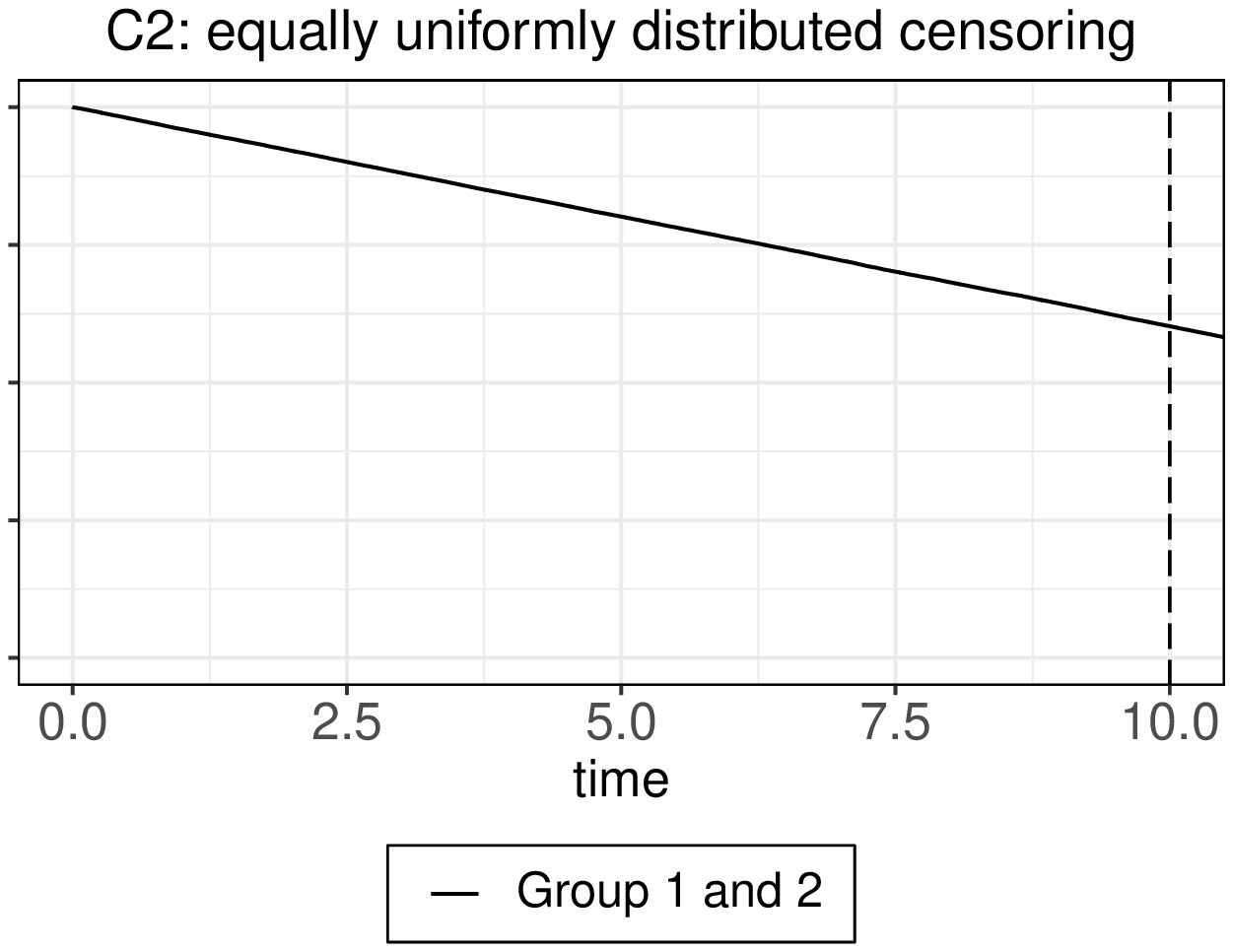}
	\end{minipage}
	\hspace{-0.3cm}
	\begin{minipage}{0.33\textwidth}
		\includegraphics[width=\textwidth]{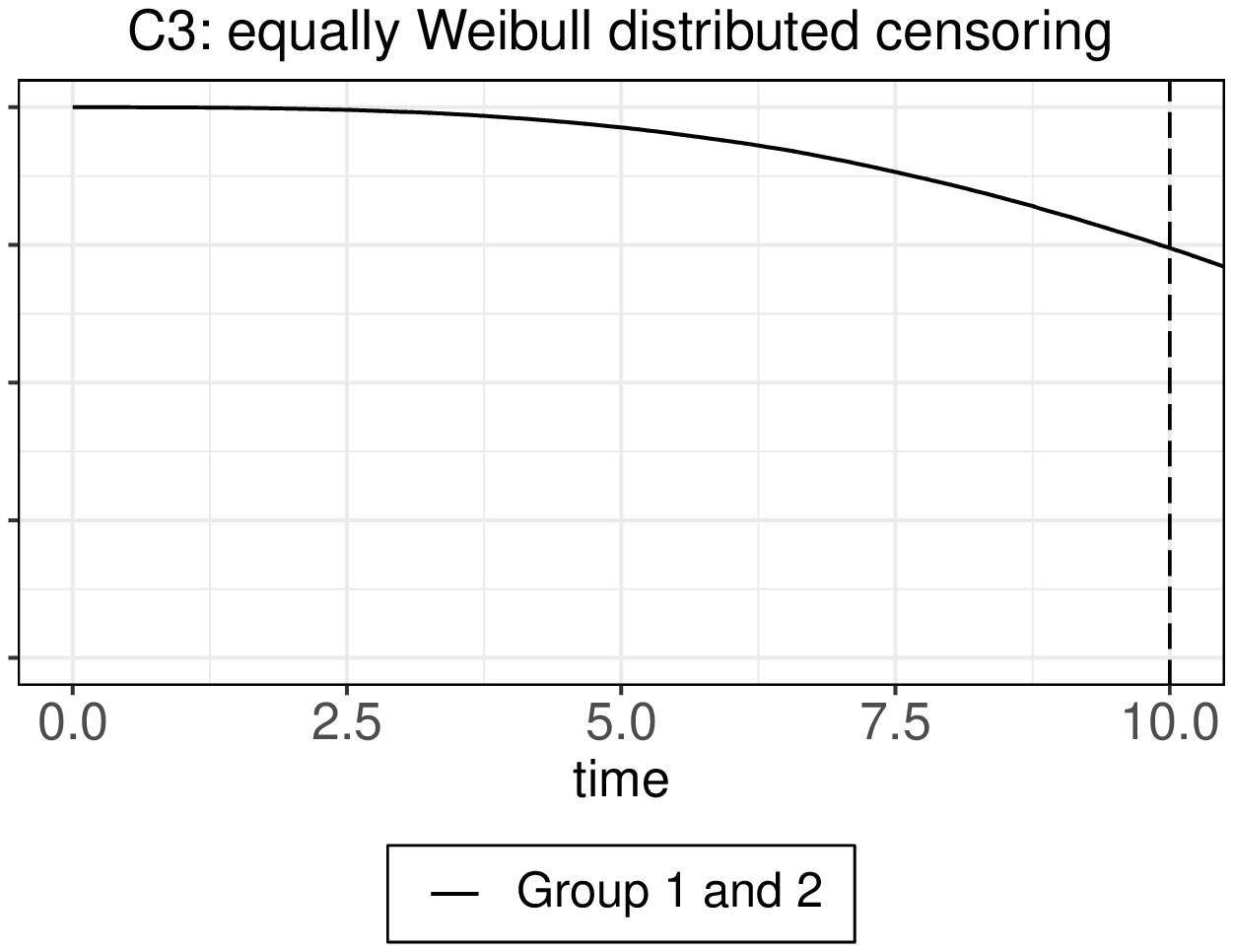}
	\end{minipage}
	\caption{The survival curves of the three different censoring scenarios.}\label{fig:survival_censoring}
\end{figure}

For the type-I error and power comparisons, we included the asymptotic test, the studentized and unstudentized permutation tests. 
While we programmed the asymptotic and the studentized permutation test by ourselves, the evaluation of the unstudentized permutation test was carried out by means of R-package \textit{survRM2perm} \citep{survRM2perm}. \cite{horiguchi2020permutation} discussed extensively different strategies on tackling the problem of possibly inestimable Kaplan--Meier-estimators for permuted data sets. Their numerical findings do not reveal a clear favorable method and all six studied strategies lead to comparable results. That is why we restricted ourselves here to the simple horizontal extension of the Kaplan--Meier curves, which corresponds to Method 2 in their paper and R-package. In detail, we set $\widehat S^\pi_i(u) = \widehat S^\pi_i(t)$ for all $u\in[t,\tau]$ when $S^\pi_i$ was just estimable up to $t<\tau$. 

The unstudentized permutation method, which relies on the assumption of exchangeable data, cannot be used to derive confidence intervals.    Consequently, just the asymptotic and studentized permutation methods were included in the respective comparisons.

The simulations were conducted by means of the computing environment {\it R} \citep{R}, version 3.6.1, generating $N_{\text{sim}}= 5,000$ simulation runs and $N_{\text{res}}=2,000$ resampling iterations for the two permutation procedures. Analogous to \cite{horiguchi2020permutation}, we regenerated the data whenever the Kaplan--Meier-estimator were not estimable, i.e. when at least for one group the largest observed time was censored and lied within $[0,\tau]$. The nominal significance level was set to  $\alpha=5\%$ and the end point of the time window was set to $\tau=10$.

\begin{table}[ht]
	\small
	\renewcommand{\arraystretch}{0.80}
	\centering
	\caption{Type-I error rates in $\%$ (nominal level $\alpha = 5\%$) for the asymptotic (Asym), the studentized permutation (st P) and the unstudentized permutation (un P) tests in Scenarios S1--S4. The values inside the binomial confidence interval [4.4$\%$, 5.6$\%$] are printed bold}\label{tab:type1_S1S4_new}
	\begin{tabular}{ccc ccc ccc ccc}
		\\ \toprule
		&&&  \multicolumn{3}{c}{$\mathbf{n} = K \cdot(24,16)$} &  \multicolumn{3}{c}{$\mathbf{n} = K \cdot(20,20)$} & \multicolumn{3}{c}{$\mathbf{n} = K \cdot(16,24)$}  \\
		\cmidrule(lr){4-6}\cmidrule(lr){7-9}\cmidrule(lr){10-12} Cens. Distr & Cens. rates & $K$ & Asym & st P & un P & Asym & st P & un P & Asym & st P & un P  \\
		\hline
		\\
		\multicolumn{12}{c}{\em S1 and S2: Exponential}\\ [6pt]
		Weib (uneq) & (7$\%$, 26$\%$) & 1 & 7.2 & \textbf{5.4} & 5.8  & 6.7 & \textbf{4.9} & \textbf{5.0}  & 6.2 & 4.3 & 4.2 \\ 
		&& 2 & 6.2 & \textbf{5.5} & 6.0 & 6.5 & \textbf{5.6} & 5.8 & \textbf{5.6} & \textbf{4.7} & \textbf{4.5} \\
		&& 4 & 6.0 & 5.7 & 6.1 & \textbf{5.3} & \textbf{4.9} & \textbf{5.1} & \textbf{5.4} & \textbf{4.9} & \textbf{4.5} \\		
		Unif (eq) & (20$\%$, 20$\%$) & 1 & 6.8 & \textbf{5.0} & \textbf{4.9} & 6.6 & \textbf{5.2} & \textbf{5.0}  & 7.0 & \textbf{4.8} & \textbf{4.9} \\ 
		&& 2 & 5.9 & \textbf{4.8} & \textbf{4.7}  & \textbf{5.5} & \textbf{4.8} & \textbf{4.8} & 6.9 & 5.8 & 5.8 \\
		&& 4 & \textbf{5.2} & \textbf{4.9} & \textbf{4.9} & \textbf{4.8} & \textbf{4.6} & \textbf{4.5} & \textbf{4.9} & \textbf{4.4} & \textbf{4.6} \\ 
		\\		
		\multicolumn{12}{c}{\em S3: Exponential vs. piece-wise Exponential (crossing curves) }\\ [6pt]
		Weib (uneq) & (7$\%$, 28$\%$) & 1 & 7.1 & \textbf{5.2} & 7.7  & 6.6 & \textbf{5.2} & \textbf{5.6}  & 6.6 & \textbf{4.6} & 4.2 \\
		&& 2 & 5.9 & \textbf{5.0} & 7.1  & 5.7 & \textbf{5.0} & \textbf{5.5} & \textbf{5.4} & \textbf{4.7} & 4.0 \\ 
		&& 4 & \textbf{4.9} & \textbf{4.5} & 7.1 & \textbf{5.2} & \textbf{5.1} & 5.8 & \textbf{4.9} & \textbf{4.6} & 3.9 \\ 	
		Unif (eq) & (20$\%$, 30$\%$) & 1 & 6.5 & \textbf{5.1} & \textbf{5.6}  & 6.6 & \textbf{5.3} & \textbf{5.2}  & 6.9 & \textbf{5.0} & 4.2 \\ 
		&& 2 & 6.3 & \textbf{5.3} & 6.4  & 6.0 & \textbf{5.2} & \textbf{5.2} & \textbf{5.6} & \textbf{4.9} & 4.0 \\ 
		&& 4 & \textbf{5.3} & \textbf{4.9} & 6.0 & 5.9 & \textbf{5.4} & \textbf{5.2} & \textbf{5.1} & \textbf{4.7} & 3.8 \\ 
		\\
		\multicolumn{12}{c}{\em S4: Lognormal}\\ [6pt]
		Weib (uneq) & (14$\%$,35$\%$) & 1 & 7.4 & \textbf{5.6} & 6.2  & 7.2 & \textbf{5.5} & 5.8  & 6.6 & \textbf{4.5} & 4.2 \\ 
		& & 2 & 6.3 & \textbf{5.4} & 6.3  & \textbf{5.5} & \textbf{4.8} & \textbf{4.9} & 5.9 & \textbf{5.0} & \textbf{4.7} \\ 
		& & 4 & \textbf{5.2} & \textbf{4.8} & 5.8 & \textbf{4.9} & \textbf{4.6} & \textbf{4.9} & \textbf{5.3} & \textbf{5.0} & \textbf{4.6} \\ 
		Unif (eq) & (33$\%$,33$\%$) & 1 & 7.1 & \textbf{5.1} & \textbf{5.4}  & 6.0 & \textbf{4.5} & 4.3  & 6.7 & \textbf{4.8} & \textbf{5.1} \\ 
		& & 2 & 5.9 & \textbf{5.1} & \textbf{5.3}  & 6.0 & \textbf{5.2} & \textbf{5.1} & 6.1 & \textbf{5.1} & \textbf{5.1} \\ 
		&& 4 & \textbf{5.4} & \textbf{4.9} & \textbf{4.9} & \textbf{5.0} & \textbf{4.7} & \textbf{4.6} & \textbf{5.2} & \textbf{4.9} & \textbf{5.1} \\ 
		\\
		\multicolumn{12}{c}{\em S5: Weibull (different shape)}\\ [6pt]
		Weib (uneq) & (8$\%$,38$\%$)  & 1 & 8.0 & 6.0 & 9.5  & 7.9 & 6.0 & 7.2  & 6.5 & \textbf{4.5} & 3.6 \\ 
		& & 2 & 6.2 & \textbf{5.3} & 8.9  & 6.6 & 5.9 & 6.9 & 6.3 & \textbf{5.4} & 4.1 \\
		& & 4 & 6.0 & \textbf{5.5} & 9.9 & \textbf{5.0} & \textbf{4.8} & 5.8 & \textbf{5.2} & \textbf{4.8} & 3.7 \\ 
		%
		Unif (eq) & (29$\%$,46$\%$) & 1 & 7.7 & 5.7 & 7.3  & 7.0 & \textbf{5.1} & \textbf{5.0}  & 6.3 & \textbf{4.7} & 3.4 \\
		& & 2 & 6.4 & \textbf{5.6} & 7.7  & 5.9 & \textbf{4.9} & \textbf{4.8} & 6.4 & \textbf{5.2} & 3.8 \\
		&& 4 & 5.7 & \textbf{5.3} & 7.1 & \textbf{5.2} & \textbf{4.9} & \textbf{4.7} & \textbf{5.0} & \textbf{4.6} & 3.2 \\ 
		\\
		\multicolumn{12}{c}{\em S6: Weibull (different scale)}\\ [6pt]
		Weib (uneq) & (8$\%$,35$\%$) & 1 & 7.9 & 6.0 & 8.8  & 7.3 & \textbf{5.6} & 6.5  & 6.3 & \textbf{4.6} & 4.0 \\ 
		& & 2 & 6.3 & \textbf{5.6} & 8.7  & 6.1 & \textbf{5.3} & 6.4 & \textbf{5.6} & \textbf{4.8} & 3.8 \\ 
		&& 4 & \textbf{5.5} & \textbf{5.1} & 8.6 & 5.7 & \textbf{5.2} & 6.4 & 5.9 & \textbf{5.5} & \textbf{4.5} \\ 
		%
		Unif (eq) & (29$\%$,35$\%$) & 1 & 7.7 & \textbf{5.6} & 6.8  & 6.7 & \textbf{5.2} & \textbf{5.2}  & 6.7 & \textbf{4.9} & 4.0 \\ 
		& & 2 & 6.6 & 5.9 & 6.9  & 5.9 & \textbf{5.2} & \textbf{5.1} & 5.8 & \textbf{5.0} & 4.0 \\
		&& 4 & \textbf{5.6} & \textbf{5.1} & 6.6 & \textbf{5.6} & \textbf{5.1} & \textbf{5.0} & \textbf{5.4} & \textbf{5.1} & 3.8 \\ 
		\\
		\multicolumn{12}{c}{\em S7: Weibull vs. piece-wise Exponential}\\ [6pt]
		Weib (uneq) & (7$\%$, 41$\%$) & 1 & 7.1 & \textbf{5.6} & 8.7  & 7.2 & \textbf{5.3} & 6.4  & 6.5 & \textbf{4.7} & 3.8 \\
		&& 2 & 6.2 & \textbf{5.3} & 9.4  & \textbf{5.5} & \textbf{5.1} & 6.1 & \textbf{5.6} & \textbf{4.9} & 3.8 \\
		&& 4 & 5.7 & \textbf{5.4} & 8.9 & \textbf{5.4} & \textbf{5.1} & 6.2 & \textbf{5.3} & \textbf{4.9} & 3.9 \\ 
		%
		Unif (eq) & (25$\%$, 47$\%$) & 1 & 7.2 & \textbf{5.5} & 6.8  & 6.6 & \textbf{5.1} & \textbf{4.6}  & 6.6 & \textbf{5.0} & 3.6 \\ 
		&& 2 & 5.7 & \textbf{5.1} & 6.4  & \textbf{5.6} & \textbf{5.1} & \textbf{5.0} & \textbf{5.6} & \textbf{4.7} & 3.5 \\
		&& 4 & 5.8 & \textbf{5.4} & 6.7 & \textbf{5.0} & \textbf{4.5} & \textbf{4.6} & \textbf{5.3} & \textbf{5.1} & 3.7 \\
		\bottomrule
	\end{tabular}
\end{table}

\subsection{Results}

The simulation results for the type-I error control are presented in Table~\ref{tab:type1_S1S4_new}. {Since the results for the two equally distributed censoring settings lead to the same conclusions, only Scenario C2 is included in the table and the results for C3 can be found in the supplement.} To judge the tests' performance, we recall that the $95\%$-confidence interval for the estimated sizes based on $N_{\text{sim}}=5,000$ simulation runs equals $[4.4\%,5.6\%]$ if the true type-I error coincides indeed with the nominal level $\alpha=5\%$. Having this at hand, it can readily be seen that the the asymptotic approach leads to rather liberal decisions. In 70 out of the 108 settings, the empirical type-I error rate was above the upper bound $5.6\%$ of the confidence interval $[4.4\%,5.6\%]$. The liberality is most pronounced under the small sample sizes settings with values up to $8.0\%$. However, the empirical sizes come closer to the nominal level when the sample sizes increased and in the majority of the largest sample size cases ($K=4$) the empirical size was inside the confidence interval $[4.4\%,5.6\%]$. These observations complement the findings of \cite{horiguchi2020permutation}, who considered only exchangeable and balanced settings. Therefore, the asymptotic test cannot be recommended for quite small sample sizes.

Switching to the unstudentized permutation test, the findings are diverse. Under Scenarios S1, S2 and S4 combined with equally distributed censoring (C2), the permutation keeps the nominal level very accurately. This observation is not surprising because these settings corresponds to a exchangeable data situation. However, this is not true anymore when unequal censoring is considered instead. While the unstudentized permutation test still controls the nominal level reasonable well for $\mathbf{n}=K(20,20)$ and $\mathbf{n}=K(16,24)$, it  exhibits a liberality for $\mathbf{n}=K\cdot(24,16)$ with empirical sizes around $6\%$. The empirical sizes under the remaining Scenarios S3, S5, S6 and S7 with crossing survival curves become even more unstable. For the equally distributed censoring setting, the unstudentized permutation test lead to rather liberal decision for $\mathbf{n}=K(24,16)$ with values up to $7.7\%$ and quite conservative decisions for $\mathbf{n}=K(16,24)$ with values reaching down to $3.2\%$. Moreover, it is apparent that the performance does not improve under these unbalance sample size settings even when the sample sizes increase. In contrast to these findings, the type-I error rate is again well preserved by the unstudentized permutation test under the balanced settings. However, this changes when we consider the censoring setting C1 with unequal censoring distributions. Here, the test exhibits a rather liberality under the Scenarios S5--S7 with values up to $7.2\%$ under the balanced sample size settings. The liberality is even more pronounced for $\mathbf{n}=K(24,16)$ with values even up to $9.5\%$. The conservativness under the other unbalanced setting, i.e. $\mathbf{n}=K(16,24)$, is now less pronounced but still present with values around $4\%$.

The overall instable type-I error performance can be explained by the systematically error mentioned in Section~\ref{sec:permutation}, which is caused by the difference between the variance of the test statistic and its permutation counterpart.  As motivated there, this can be fixed by studentization. The studentized permutation tests keep the type-I error rate in almost all settings inside the binomial confidence interval and have slight deviations outside. Overall it leads to the most stable results under the null hypothesis and we recommend its application whenever the sample sizes are rather small (e.g. $n_1+n_2<100$).

Due to limited space, the results of the power comparisons are deferred to the supplement. We summarize the findings of the comparison as follows. In most of the cases, the asymptotic tests leads to the highest power values, where the difference in power to the studentized permutation test even tend up to $4-5$ percentage points. The differences are most pronounced for the smaller sample size settings ($K=1,2$) and can be explained by the liberal behavior of the asymptotic test, which we observed under the null hypotheses.  Comparing the power results of the two permutation approaches, the power values are almost indistinguishable in most of the cases. However, partially the unstudentized permutation test lead to higher power values with a difference up to even $6$ percentage points and even the reverse, i.e. the studentized permutation has higher power values, can be observed. These diverse findings can be explained by the unstable type-I error control of the unstudentized permutation test with too liberal and too conservative decisions. Overall, the results need to be taken with a pinch of salt, because only the studentized permutation test exhibited a generally convincing performance under the null hypotheses.

We finally turn to the performance of the confidence intervals. We summarized the results for all seven distributional choices S1--S7, the three censoring distributions C1--C3 and the five different choices for $\delta\in\{0,0.5,1,1.5,2\}$ in Figure~\ref{fig:CI_diff}, for each of the nine different sample sizes. In total, each boxplot summarizes the results of 102 different settings; recall that S1 and S2 coincide under $\delta = 0$ and, thus, only S1 is considered in this case. It is apparent that the empirical coverage of the asymptotic test is liberal, similar to our findings regarding the type-I error control. The liberality or undercoverage is most pronounced for the small sample size cases $(K=1)$  and becomes less pronounced when the sample sizes increase. But even for the largest sample size settings $(K=4)$, the median empirical coverage is just slightly above the lower border $94.4\%$ of the binomial $95\%$-confidence interval $[94.4\%,95.6\%]$. In contrast, the permutation-based confidence intervals lead to more satisfactory results for all considered sample sizes and all boxes are clearly inside the $95\%$-confidence interval $[94.4\%,95.6\%]$, except for $\mathbf{n} =(24,16)$, where the lower end of box is slightly outside the confidence interval.

\begin{figure}
	\centering
	\includegraphics[width=\textwidth]{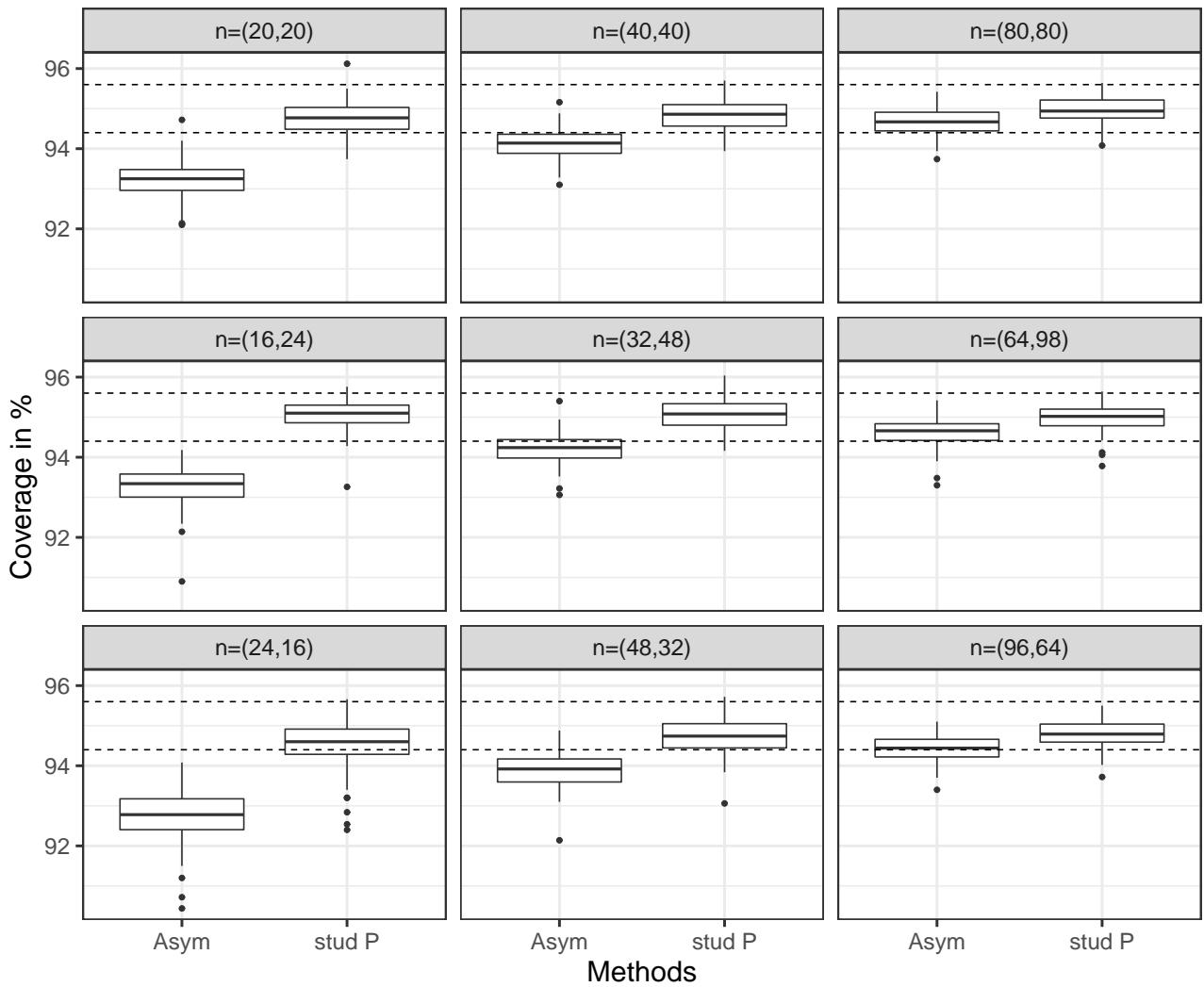}
	\caption{Coverage in $\%$ (nominal level $\alpha=5\%$) of the confidence intervals based on the asymptotic approximation (Asym) and the  studentized permutation approach (stud P). The dashed, horizontal lines represent the binomial 95$\%$-confidence interval $[94.4\%,95.6\%]$} \label{fig:CI_diff}
\end{figure}

In summary, we can only recommend the studentized permutation test and the corresponding permutation-based confidence intervals for the quantity $\mu_1-\mu_2$ for small sample sizes, as it leads to the most accurate type-I error and coverage control, respectively. Moreover, it can compete in terms of power with the other strategies whenever a comparison is fair and not influenced by liberal decisions under the null hypothesis.

\section{Real data example}\label{sec:data_example}

\begin{figure}
	\centering
	\begin{minipage}{0.9\textwidth}
		\includegraphics[width=\textwidth]{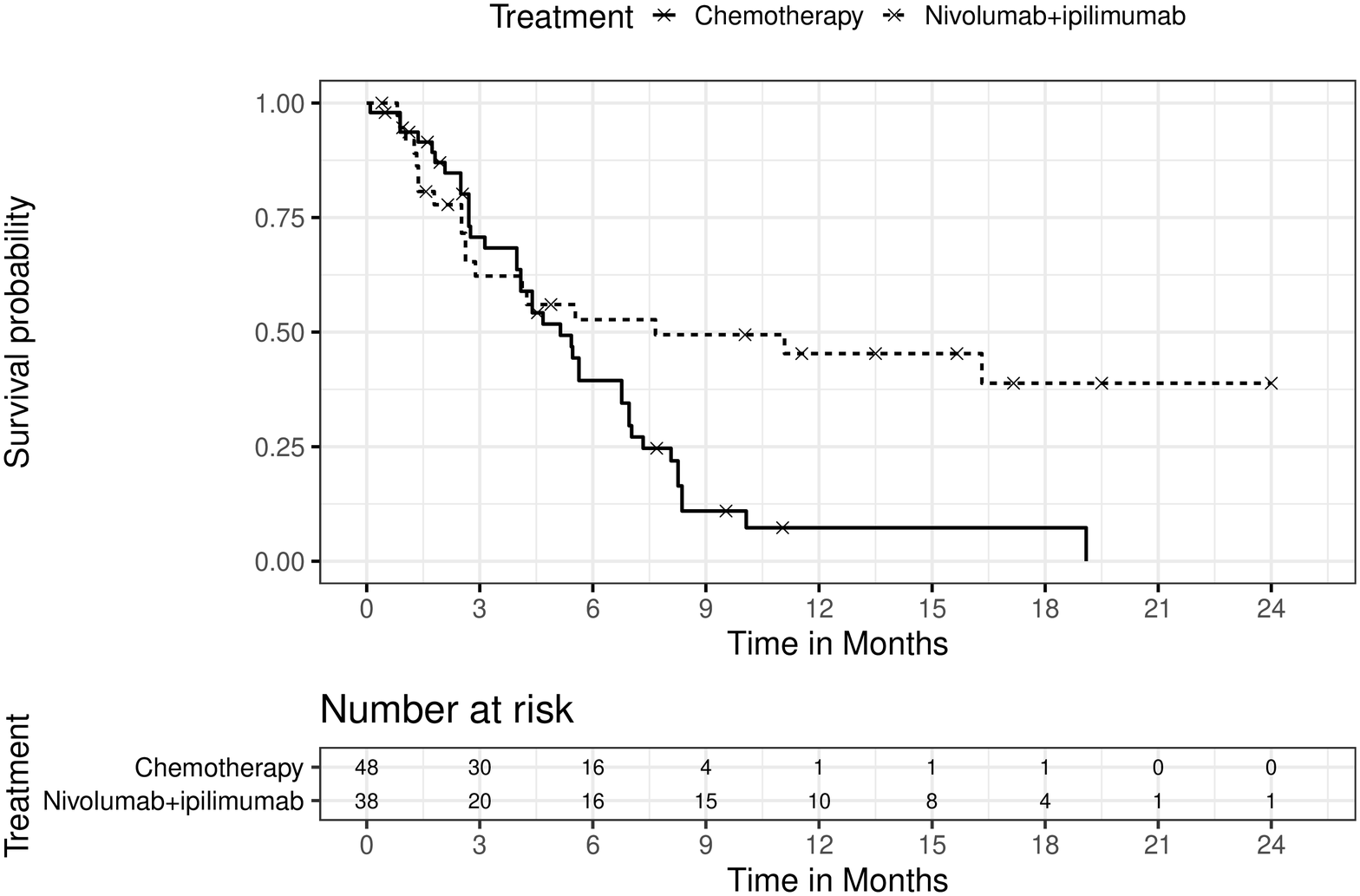}
	\end{minipage}
	\caption{Kaplan--Meier curves of the reconstructed data}\label{fig:data_example}
\end{figure}

To illustrate the presented permutation-based methods, we re-consider the data analysis of \cite{hellmann2018nivolumab}, who compared a combination treatment of nivolumab plus ipilimumab with chemotherapy among $299$ patients with non-small-cell lung cancer. Their study focused on patients with a high tumor mutational burden, i.e. at least ten mutations per megabase. And the study endpoint was progression-free survival. Since the present methods are designed for small sample sizes, we conduct a relevant subgroup analysis, which was also done by \cite{hellmann2018nivolumab}. In detail, we restrict to the patients having PD-L1 (tumor programmed death
ligand 1) expression of at least $1\%$. On the basis of the published Kaplan--Meier curves in \cite{hellmann2018nivolumab} and some additional information therein, e.g. the risk table, we reconstructed the individual patient data following the procedure of \cite{guyot2012enhanced}. The respective Kaplan--Meier curves of the two treatment groups are displayed in Figure~\ref{fig:data_example}. Therein, we can observe a delayed treatment effect of nivolumab plus ipilimumab. Thus, the assumption of proportional hazards is questionable and can even by formally rejected by the well established test of \cite{grambsch1994proportional} or the recent permutation-based proposal of \cite{ditzhaus2020bootstrap} (with 10,000 permutations). Both tests lead to a $p$-value less than $0.1\%$. That is why the original analysis of \cite{hellmann2018nivolumab} using the hazard ratio (HR: $0.48$ and $95\%$ CI $[0.27,0.85]$) need to be considered carefully. The RMSTs offer the possibility to interpret the treatment effect easily beyond the Cox model. The $p$-values of the asymptotic, studentized and unstudentized permutation tests (both based on 5,000 permutations), for inferring $\mathcal H_0:\mu_1=\mu_2$ are presented in Table~\ref{tab:data+examp+pval}. The confidence intervals for the difference $\mu_2-\mu_1$ as well as for the ratio $\mu_1/\mu_2$ are shown in Table~\ref{tab:data+examp+conf}. In both tables, the different end points $\tau\in\{12,15,18\}$ were considered. In practice, the end point needs to be chosen jointly with the physician regarding clinical relevance.

\begin{table}[t]
	\centering
	\caption{Testing RMST difference based on the asymptotic (Asym), the studentized (st P) and unstudentized (un P) tests for the reconstructed data}
	\label{tab:data+examp+pval}
	\begin{tabular}{cccccccccc}
		\\
		\toprule
		\\ [-10pt]
		&\multicolumn{3}{c}{$\tau = 12$ months} & \multicolumn{3}{c}{$\tau = 15$ months} & \multicolumn{3}{c}{$\tau = 18$ months} \\
		\cmidrule(lr){2-4}\cmidrule(lr){5-7}\cmidrule(lr){8-10}& Asym & un P & st P & Asym & un P & st P & Asym & un P & st P  \\
		\hline
		\\ [-6pt]
		$p$-values & 0.045 & 0.045 & 0.067 & 0.01 & 0.011 & 0.02 & 0.004 & 0.005 & 0.011\\
		\bottomrule
	\end{tabular}
	
\end{table}

\begin{table}[t]
	\centering
	\caption{Point estimates and $95\%$-confidence intervals of the difference $\mu_1-\mu_2$ and the ratio $\mu_1/\mu_2$, respectively, based on the asymptotic approximation (Asym) and the studentized permutation method. The first group is the chemotherapy group and the second the nivolumab plus ipilimumab group} \label{tab:data+examp+conf}
	\begin{tabular}{lcccccc}
		\\
		\toprule
		\\ [-10pt]
		&  \multicolumn{2}{c}{$\tau = 12$ months} & \multicolumn{2}{c}{$\tau = 15$ months} & \multicolumn{2}{c}{$\tau = 18$ months} \\
		\cmidrule(lr){2-3}\cmidrule(lr){4-5}\cmidrule(lr){6-7} & Asym & st P  & Asym & st P  & Asym & st P   \\
		\hline \\[-6pt]
		$\widehat\mu_{1}-\widehat\mu_{2}$ & \multicolumn{2}{c}{-1.85} & \multicolumn{2}{c}{-2.99} & \multicolumn{2}{c}{-4.02} \\ [6pt]
		$95\%$-CI &  $[\text{-}3.66,\text{-}0.04]$ & $[\text{-}3.83,0.13]$ & [$\text{-}5.28,\text{-}0.70 $] &  $[\text{-}5.53,\text{-}0.43]$ & $[\text{-}6.79,\text{-}1.26]$ & $[\text{-}7.09,\text{-}0.96]$ \\
		\\
		$\widehat\mu_{1}/\widehat\mu_{2}$ &  \multicolumn{2}{c}{0.75} & \multicolumn{2}{c}{0.65} & \multicolumn{2}{c}{0.59} \\ [6pt]
		$95\%$-CI & [0.57,0.98] & [0.56,1.01] & [0.48,0.88] & [0.47,0.91] & [0.43, 0.82] & [0.41,0.85]\\
		\bottomrule
	\end{tabular}
\end{table}

For $\tau=15$ and $\tau=18$, the results confirm the findings of \cite{hellmann2018nivolumab} that the combination nivolumab plus ipilimumab improves the progression-free time compared to the chemotherapy. The point estimates and confidence intervals in Table~\ref{tab:data+examp+conf} help to quantify the improvement and can be interpreted easily. For example, the combination treatment leads in average to a longer progression-free time of $4.02\pm 3.05$ months ($95\%$ confidence based on 5,000 permutations) compared to the chemotherapy over the first 1.5 years.

In general, it is observed that the asymptotic approach leads to smaller $p$-values and narrower confidence intervals than its permutation counterpart. Moreover, the unstudentized permutation test lead to comparable $p$-values than the asymptotic approach. As pointed out in Section~\ref{sec:sim}, the results of the asymptotic and unstudentized permutation test need to be considered carefully, especially for small and unbalanced sample sizes as having here. Thus, we would rather trust the results of the studentized permutation test than those of the other two, especially for $\tau=12$ months, where the decisions are diverse.

\section{Discussion and remarks}\label{sec:discussion}

In the last years, the RMST became an important part of the statistical toolbox for survival data. Various researchers \citep{stensrud2020test,trinquart2016comparison,a2016restricted} advise to use it, at least, as a complementary summary statistic, especially when the assumption of proportional hazards is in doubt. 
As raised by \cite{horiguchi2020permutation}, the type-I error rate of related asymptotic methods is inflated for small sample sizes. The permutation procedure of \cite{horiguchi2020permutation} as well as their detailed discussion of how to deal with inestimable Kaplan--Meier curves of the permutated data was an important step to solve that problem. However, their test's application is limited to exchangeable data settings and, in particular, to equal survival and censoring distributions, respectively. 

In this paper, we explained how studentization can tackle these limitations.  For the present survival two-sample comparison, it allows us to apply permutation tests even in non-exchangeable data situation, i.e. for different survival and/or censoring distributions, as well as to formulate corresponding confidence intervals for the quantity $\mu_1-\mu_2$ and $\mu_1/\mu_2$ of interest.
%
%
Moreover, the control of the type-I error, which was the initial motivation for permutation tests, is not affected by the studentization strategy. Compared to their asymptotic counterparts, studentized permutation tests usually show a satisfactory type-I error control even for small sample sizes, as seen in Section~\ref{sec:sim}.

The theoretical justification of studentized permutation tests and respective confidence intervals is complemented by an extensive simulation study. The corresponding results support the usage of the developed methods for small sample sizes.

Our framework can be extended in various directions, e.g. to competing risks \citep{zhao2018estimating,lyu2020use}. More general study designs may be part of future research. For that purpose, we can follow \cite{dobler2020factorial} and \cite{ditzhaus2020inferring}, who recently discussed permutation-based inference for the concordance measure and median survival times, respectively, in the general context of factorial designs. Sample size determination can also be developed, in parallel to the asymptotic test based results \citep{TingYu2018}.

\section*{Acknowledgement}
Marc Ditzhaus was funded by the \textit{Deutsche Forschungsgemeinschaft} (grant no.  PA-2409 5-1). Moreover, the authors gratefully acknowledge the computing time provided on the Linux HPC cluster at TU Dortmund (LiDO3), partially funded in the course of the Large-Scale Equipment Initiative by the \textit{Deutsche Forschungsgemeinschaft} as project 271512359.


\bibliographystyle{plainnat}
\bibliography{sample}

\newpage

\appendix

\section{Additional simulation results}
First, we present the results for the power comparison, see Tables~\ref{tab:power_exp}--\ref{tab:power_weibull}, and for the type-I error comparison in Table~\ref{tab:type1_supplement} under the remaining censoring setting C3, i.e. equally Weibull distributed censoring. The results were already briefly discussed in the main paper, including all relevant main conclusions. In addition to that, we present here the simulation results of the asymptotic and permutation-based confidence intervals for the ratio $\mu_1/\mu_2$ from Section~\ref{sec:ratio}. For the respective simulation study, we used the same set-up as in the simulation study for the difference-based methods from Section~\ref{sec:sim}. In particular, each boxplot in Figure~\ref{fig:CI_rat} summarizes the results for $102$ different settings. It is apparent that the performance of the asymptotic confidence interval is less extreme than the one for the differences. But for small sample sizes, here $\mathbf{n}=(20,20),(16,24),(24,16)$, it still leads to an undercoverage or liberal decisions. For the moderate sample sizes, the undercoverage almost vanishes and, for the large sample sizes, it completely vanishes with just a few exceptions. The permutation-based confidence intervals lead to more stable results, especially under $K\mathbf{n}_{\text{bal}},K\mathbf{n}_{\text{incr}},K\mathbf{n}_{\text{decr}}$ with $K=1,2$ and, thus, it is still our recommendation for sample sizes $n_1+n_2<100$.

\begin{table}[ht]
	\small
	\renewcommand{\arraystretch}{0.80}
	\centering
	\caption{Type-I error rates in $\%$ (nominal level $\alpha = 5\%$) for the asymptotic (Asym), the studentized permutation (st P) and the unstudentized permutation (un P) tests in Scenarios S1--S7 under equally Weibull distributed censoring, i.e. censoring setting C3. The values inside the binomial confidence interval [4.4$\%$, 5.6$\%$] are printed bold. }\label{tab:type1_supplement}
	\begin{tabular}{Hcc ccc ccc ccc}
		\\ \toprule 
		&&&  \multicolumn{3}{c}{$\mathbf{n} = K \cdot(24,16)$} &  \multicolumn{3}{c}{$\mathbf{n} = K \cdot(20,20)$} & \multicolumn{3}{c}{$\mathbf{n} = K \cdot(16,24)$}  \\
		\cmidrule(lr){4-6}\cmidrule(lr){7-9}\cmidrule(lr){10-12} 	Cens. Distr & Cens. rates & $K$ & Asym & st P & un P & Asym & st P & un P & Asym & st P & un P  \\
		\hline
		\\
		\multicolumn{12}{c}{\em S1 and S2: Exponential}\\ [6pt]
		Weib (eq) & (11$\%$, 11$\%$) & 1 & 6.1 & \textbf{4.6} & \textbf{4.7}  & 6.4 & \textbf{5.0} & \textbf{5.0}  & 5.9 & \textbf{4.7} & \textbf{4.7} \\ 
		&& 2 & 5.8 & \textbf{4.9} & \textbf{5.0}  & 5.7 & \textbf{5.0} & \textbf{5.0} & \textbf{5.6} & \textbf{4.6} & \textbf{4.8} \\ 
		&& 4 & \textbf{5.5} & \textbf{5.0} & \textbf{5.1} & \textbf{4.7} & \textbf{4.4} & \textbf{4.5} & \textbf{5.4} & \textbf{5.0} & \textbf{5.0} \\ 
		\\		
		\multicolumn{12}{c}{\em S3: Exponential vs. piece-wise Exponential (crossing curves) }\\ [6pt]
		Weib (eq) & (11$\%$, 27$\%$) & 1 & 7.7 & 5.9 & 7.2  & 6.3 & \textbf{4.9} & \textbf{5.0}  & 6.4 & \textbf{4.7} & 3.9 \\ 
		&& 2 & 6.1 & \textbf{5.5} & 6.7  & \textbf{5.3} & \textbf{4.8} & \textbf{4.9} & 5.9 & \textbf{5.1} & 3.8 \\ 
		&& 4 & 5.7 & \textbf{5.3} & 6.6 & 5.8 & \textbf{5.4} & \textbf{5.4} & 6.1 & 5.9 & \textbf{4.7} \\
		\\
		\multicolumn{12}{c}{\em S4: Lognormal}\\ [6pt]
		Weib (eq) & (21$\%$,21$\%$) & 1 & 6.3 & \textbf{5.1} & \textbf{4.8}  & 5.9 & \textbf{4.6} & \textbf{4.8}  & 7.5 & \textbf{5.5} & \textbf{5.6} \\
		& & 2 & 5.8 & \textbf{4.9} & \textbf{5.2}  & \textbf{5.4} & \textbf{4.7} & \textbf{4.8} & \textbf{5.3} & \textbf{4.7} & \textbf{4.7} \\ 
		&& 4 & \textbf{5.2} & \textbf{4.7} & \textbf{4.8} & \textbf{5.4} & \textbf{5.1} & \textbf{5.1} & \textbf{5.3} & \textbf{5.0} & \textbf{4.9} \\ 
			\\
			\multicolumn{12}{c}{\em S5: Weibull (different shape)}\\ [6pt]
			Weib (eq) & (13$\%$,40$\%$) & 1 & 7.5 & 5.7 & 7.7  & 7.2 & 5.7 & \textbf{5.6}  & 6.2 & \textbf{4.6} & 3.2 \\ 
			& & 2 & 5.9 & \textbf{4.9} & 7.1  & \textbf{5.6} & \textbf{5.2} & \textbf{4.9} & 5.8 & \textbf{5.0} & 3.5 \\ 
			&& 4 & \textbf{5.3} & \textbf{5.0} & 6.9 & \textbf{5.1} & \textbf{4.7} & \textbf{4.9} & \textbf{5.3} & \textbf{5.0} & 3.3 \\
			\\
			\multicolumn{12}{c}{\em S6: Weibull (different scale)}\\ [6pt]
			Weib (eq) & (13$\%$,26$\%$) & 1 & 6.6 & \textbf{4.9} & 6.0 & 6.1 & \textbf{4.8} & \textbf{4.8}  & 6.4 & \textbf{4.8} & 3.7 \\ 
			& & 2 & 6.6 & 5.9 & 7.0 & \textbf{5.3} & \textbf{4.6} & \textbf{4.6} & \textbf{5.5} & \textbf{4.6} & 3.5 \\ 
			&& 4 & \textbf{5.3} & \textbf{5.0} & 6.4 & \textbf{5.3} & \textbf{5.0} & \textbf{5.1} & \textbf{5.0} & \textbf{4.9} & 3.6 \\ 
			\\
			\multicolumn{12}{c}{\em S7: Weibull vs. piece-wise Exponential}\\ [6pt]
			%
			%
			Weib (eq) & (11$\%$, 43$\%$) & 1 & 6.5 & \textbf{4.7} & 6.4  & 7.1 & \textbf{5.6} & \textbf{5.6}  & 6.2 & \textbf{4.9} & 3.5 \\
			&& 2 & 6.7 & 6.1 & 7.5  & 5.9 & \textbf{5.2} & \textbf{5.2} & \textbf{5.3} & \textbf{4.5} & 3.0 \\ 
			&& 4 & \textbf{4.9} & \textbf{4.6} & 6.3 & \textbf{5.2} & \textbf{5.0} & \textbf{5.1} & \textbf{5.2} & \textbf{4.8} & 3.7 \\ 
			\bottomrule
	\end{tabular}
\end{table}

\begin{table}[ht!]
	\small
	\renewcommand{\arraystretch}{0.80}
	\caption{Power values in $\%$ (nominal level $\alpha = 5\%$) under the alternative $\mu_2 -\mu_1 = \delta \in\{1,2\}$ for the asymptotic (Asym), the studentized permutation (st P) and the unstudentized permutation (un P) tests in Scenarios S1, S2 and S3.}  \label{tab:power_exp}
	\centering
	\begin{tabular}{ccccccccccccc}
		\\ \toprule
		&&&&  \multicolumn{3}{c}{$\mathbf{n} = K \cdot(24,16)$} & \multicolumn{3}{c}{$\mathbf{n} = K \cdot(16,24)$} &  \multicolumn{3}{c}{$\mathbf{n} = K \cdot(20,20)$} \\
	\cmidrule(lr){5-7}\cmidrule(lr){8-10}\cmidrule(lr){11-13} 	Cens. Distr & $\delta$ & Cens. rates & $K$ & Asym & st P & un P & Asym & st P & un P & Asym & st P & un P  \\
		\hline
		\\
		\multicolumn{13}{c}{\em S1: Exponential (proportional hazards)}\\ [6pt]
		Weib (uneq) & $\delta =1$ & (7$\%$,30$\%$) & 1 & 15.3 & 12.2 & 15.1 & 15.3 & 12.5 & 13.6 & 16.1 & 12.6 & 12.2 \\ 
		&& & 2 & 23.3 & 21.4 & 24.6 & 23.8 & 21.8 & 22.9 & 23.8 & 21.6 & 20.6 \\
		&& & 4 & 38.9 & 37.4 & 41.7 & 41.2 & 40.0 & 41.4 & 41.1 & 39.4 & 38.1 \\ 
		& $\delta = 2$ & (7$\%$,34$\%$) & 1 & 40.6 & 35.2 & 40.2 & 42.3 & 37.7 & 39.0 & 42.8 & 37.0 & 36.2 \\ 
		&& & 2 & 64.0 & 61.4 & 65.9 & 69.2 & 67.1 & 68.3 & 67.3 & 65.0 & 63.9 \\ 
		&& & 4 & 90.1 & 89.4 & 91.4 & 92.5 & 92.2 & 92.5 & 92.4 & 91.9 & 91.2 \\ 
		Unif (eq) & $\delta =1$ & (20$\%$,27$\%$) & 1 & 16.0 & 12.3 & 13.7 & 16.3 & 13.8 & 13.8 & 18.1 & 14.2 & 13.3 \\ 
		&& & 2 & 24.1 & 22.0 & 23.3 & 25.7 & 23.6 & 23.5 & 25.3 & 23.1 & 22.3 \\ 
		&& & 4 & 41.3 & 40.3 & 41.4 & 41.5 & 40.4 & 40.4 & 42.4 & 41.4 & 40.0 \\ 
		& $\delta =2$ & (20$\%$,37$\%$) & 1 & 41.7 & 35.4 & 38.0 & 43.6 & 38.5 & 38.7 & 42.2 & 36.9 & 36.1 \\ 
		&& & 2 & 65.8 & 63.5 & 64.9 & 69.9 & 68.0 & 67.9 & 67.8 & 65.2 & 64.7 \\ 
		&& & 4 & 91.2 & 90.6 & 91.1 & 92.8 & 92.4 & 92.3 & 91.6 & 91.2 & 91.1 \\ 
		Weib (eq) & $\delta =1$ & (11$\%$,19$\%$) & 1 & 16.8 & 13.5 & 15.0 & 16.1 & 13.7 & 13.7 & 17.5 & 14.5 & 13.7 \\ 
		&&& 2 & 26.0 & 24.1 & 25.3 & 26.6 & 24.8 & 24.9 & 25.3 & 23.3 & 22.8 \\
		&& & 4 & 41.6 & 40.8 & 41.7 & 46.1 & 44.8 & 44.9 & 44.6 & 43.3 & 42.4 \\ 
		& $\delta =2$ & (11$\%$,29$\%$) & 1 & 44.9 & 40.3 & 42.1 & 47.0 & 42.0 & 42.1 & 44.9 & 40.1 & 39.2 \\ 
		&& & 2 & 71.6 & 69.4 & 70.9 & 72.4 & 70.6 & 70.3 & 72.2 & 69.6 & 69.0 \\ 
		&& & 4 & 93.4 & 93.0 & 93.5 & 95.1 & 95.0 & 95.0 & 94.2 & 93.8 & 93.4 \\		
		\\
		\multicolumn{13}{c}{\em S2: Exponential (late departures)}\\ [6pt]
		Weib (uneq) & $\delta =1$ & (7$\%$,30$\%$) & 1 & 13.8 & 10.9 & 14.2 & 14.8 & 12.1 & 13.3 & 16.2 & 12.8 & 11.5 \\
		&& & 2 & 20.9 & 19.1 & 23.8 & 21.7 & 20.0 & 21.4 & 23.2 & 21.2 & 19.5 \\ 
		&& & 4 & 36.1 & 34.7 & 40.2 & 39.6 & 38.5 & 40.3 & 38.9 & 37.9 & 35.5 \\ 
		& $\delta = 2$ & (7$\%$,39$\%$) & 1 & 34.8 & 30.2 & 36.6 & 37.8 & 32.8 & 34.7 & 38.1 & 32.7 & 30.1 \\ 
		&& & 2 & 56.1 & 53.2 & 60.3 & 61.7 & 59.4 & 61.2 & 62.7 & 59.6 & 57.2 \\ 
		&& & 4 & 84.8 & 84.0 & 87.9 & 88.3 & 87.7 & 88.6 & 87.7 & 87.1 & 85.6 \\ 
		Unif (eq) & $\delta =1$ & (20$\%$,31$\%$) & 1 & 14.8 & 11.2 & 13.1 & 15.1 & 12.3 & 12.3 & 15.8 & 12.9 & 11.1 \\ 
		&& & 2 & 22.0 & 20.3 & 22.2 & 23.8 & 21.8 & 22.2 & 24.2 & 22.0 & 19.9 \\ 
		&& & 4 & 36.5 & 35.3 & 37.3 & 39.4 & 38.2 & 38.1 & 37.8 & 36.6 & 34.1 \\ 
		& $\delta =2$ & (20$\%$,49$\%$) & 1 & 36.8 & 31.3 & 34.5 & 38.3 & 33.9 & 33.4 & 38.6 & 33.5 & 30.7 \\ 
		&& & 2 & 59.5 & 56.9 & 59.8 & 62.5 & 60.1 & 59.8 & 61.8 & 59.4 & 56.2 \\ 
		&& & 4 & 85.7 & 84.8 & 86.4 & 88.6 & 87.9 & 88.0 & 89.4 & 88.6 & 86.6 \\ 
		Weib (eq) & $\delta =1$ & (11$\%$,24$\%$) & 1 & 15.4 & 12.5 & 14.6 & 16.2 & 13.4 & 13.3 & 15.7 & 12.8 & 11.5 \\
		&&& 2 & 22.6 & 20.6 & 23.0 & 24.3 & 22.3 & 22.6 & 24.8 & 23.0 & 20.6 \\
		&& & 4 & 38.8 & 37.5 & 40.0 & 42.1 & 41.1 & 41.0 & 42.2 & 40.9 & 38.7 \\ 
		& $\delta =2$ & (11$\%$,46$\%$) & 1 & 38.4 & 33.8 & 37.2 & 41.0 & 36.8 & 36.8 & 40.8 & 36.4 & 33.5 \\ 
		&& & 2 & 62.4 & 60.1 & 63.3 & 65.8 & 63.9 & 63.8 & 66.4 & 64.1 & 61.1 \\ 
		&& & 4 & 88.3 & 87.8 & 89.4 & 90.7 & 90.0 & 90.1 & 91.2 & 90.7 & 89.0 \\		
		\\
		\multicolumn{13}{c}{\em S3: Exponential vs. piece-wise Exponential (crossing curves) }\\ [6pt]
		Weib (uneq) & $\delta =1$ & (7$\%$,32$\%$) & 1 & 14.3 & 11.5 & 15.1 & 14.3 & 11.8 & 12.6 & 15.1 & 12.0 & 10.0 \\ 
		&& & 2 & 18.7 & 16.9 & 21.9 & 21.2 & 19.4 & 21.1 & 21.6 & 19.5 & 17.3 \\
		&& & 4 & 31.0 & 29.8 & 37.1 & 32.6 & 31.7 & 33.2 & 35.3 & 34.2 & 30.9 \\ 
		& $\delta=2$ & (7$\%$,36$\%$) & 1 & 36.3 & 31.4 & 37.3 & 38.6 & 33.6 & 35.1 & 38.9 & 33.4 & 30.9 \\
		&& & 2 & 57.7 & 54.9 & 61.9 & 61.1 & 58.6 & 60.0 & 61.7 & 59.2 & 55.8 \\
		&& & 4 & 85.1 & 84.7 & 87.5 & 87.5 & 86.8 & 87.5 & 88.6 & 88.0 & 86.3 \\
		Unif  (eq)& $\delta=1$ & (20$\%$,38$\%$) & 1 & 13.7 & 10.8 & 12.5 & 14.3 & 11.4 & 11.2 & 15.8 & 13.1 & 10.8 \\
		&&  & 2 & 20.7 & 18.8 & 21.5 & 20.5 & 18.6 & 18.4 & 21.6 & 19.7 & 16.4 \\ 
		&&  & 4 & 34.1 & 32.8 & 36.5 & 34.6 & 33.6 & 33.2 & 36.2 & 34.9 & 30.9 \\
		& $\delta=2$ & (20$\%$,46$\%$) & 1 & 36.8 & 31.4 & 34.9 & 38.6 & 33.9 & 33.7 & 38.2 & 33.4 & 30.6 \\ 
		& &  & 2 & 59.8 & 56.8 & 60.1 & 62.2 & 59.7 & 59.6 & 62.5 & 59.7 & 56.7 \\
		&& & 4 & 86.5 & 85.7 & 87.5 & 89.3 & 88.8 & 88.5 & 88.5 & 87.7 & 86.0 \\ 
		Weib (eq) & $\delta=1$ & (11$\%$,35$\%$) & 1 & 14.3 & 11.9 & 14.0 & 14.9 & 12.6 & 12.6 & 14.7 & 11.9 & 9.8 \\ 
		&&& 2 & 20.9 & 19.1 & 22.0 & 22.3 & 20.8 & 20.9 & 23.4 & 21.8 & 18.3 \\ 
		&& & 4 & 33.5 & 32.7 & 36.7 & 38.5 & 37.7 & 37.7 & 37.3 & 36.2 & 32.4 \\ 
		& $\delta=2$ & (11$\%$,42$\%$) & 1 & 38.0 & 33.0 & 36.7 & 39.5 & 35.3 & 35.4 & 41.5 & 36.5 & 33.8 \\ 
		& && 2 & 61.6 & 58.9 & 62.6 & 64.2 & 62.0 & 62.3 & 65.5 & 63.0 & 59.6 \\
		&& & 4 & 88.5 & 87.9 & 89.9 & 91.3 & 90.9 & 90.8 & 90.4 & 89.9 & 88.2 \\ 
		\bottomrule
	\end{tabular}
\end{table}

\begin{table}[ht!]
	\small
	\renewcommand{\arraystretch}{0.8}
	\caption{Power values in $\%$ (nominal level $\alpha = 5\%$) under the alternative $\mu_2 -\mu_1 = \delta \in\{1,2\}$ for the asymptotic (Asym), the studentized permutation (st P) and the unstudentized permutation (un P) tests in Scenarios S4 and S5.}\label{tab:power_lognormal}	
	\centering
	\begin{tabular}{ccccccccccccc}
		\\ \toprule
		&&&&  \multicolumn{3}{c}{$\mathbf{n} = K \cdot(24,16)$} & \multicolumn{3}{c}{$\mathbf{n} = K \cdot(16,24)$} &  \multicolumn{3}{c}{$\mathbf{n} = K \cdot(20,20)$} \\
		\cmidrule(lr){5-7}\cmidrule(lr){8-10}\cmidrule(lr){11-13} Cens. Distr & $\delta$ & Cens. rates & $K$ & Asym & st P & un P & Asym & st P & un P & Asym & st P & un P  \\
		\hline
		\\
		\multicolumn{13}{c}{\em S4: Lognormal (scale alternatives)}\\ [6pt]
		Weib (uneq) & $\delta =1$ & (14$\%$,39$\%$) & 1 & 30.2 & 25.5 & 24.2 & 28.4 & 24.1 & 23.1 & 27.3 & 22.5 & 22.5 \\ 
		&& & 2 & 45.7 & 42.9 & 43.5 & 45.3 & 42.8 & 42.2 & 45.1 & 41.5 & 40.9 \\ 
		&& & 4 & 70.6 & 69.5 & 70.6 & 72.0 & 71.1 & 71.3 & 71.8 & 70.6 & 69.8 \\ 
		& $\delta = 2$ & (12$\%$,44$\%$) & 1 & 83.0 & 79.1 & 76.7 & 84.5 & 80.8 & 79.2 & 82.5 & 76.9 & 77.9 \\ 
		&& & 2 & 98.2 & 97.7 & 97.6 & 98.5 & 98.2 & 98.0 & 98.1 & 97.6 & 97.6 \\ 
		&& & 4 & 100.0 & 100.0 & 100.0 & 100.0 & 100.0 & 100.0 & 100.0 & 100.0 & 100.0 \\ 
		Unif (eq) & $\delta =1$ & (33$\%$,42$\%$) & 1 & 27.3 & 22.1 & 20.6 & 26.5 & 22.2 & 22.1 & 25.5 & 20.4 & 22.9 \\ 
		&& & 2 & 42.7 & 39.9 & 38.1 & 43.2 & 40.6 & 40.3 & 42.5 & 38.9 & 41.2 \\ 
		&& & 4 & 69.9 & 68.7 & 67.6 & 71.6 & 70.5 & 70.3 & 68.6 & 67.3 & 68.8 \\ 
		& $\delta =2$ & (33$\%$,56$\%$) & 1 & 81.7 & 77.2 & 74.2 & 81.4 & 77.3 & 77.5 & 76.6 & 70.2 & 76.7 \\ 
		&& & 2 & 97.7 & 97.1 & 96.5 & 97.8 & 97.4 & 97.5 & 96.4 & 95.5 & 96.8 \\
		&& & 4 & 100.0 & 100.0 & 100.0 & 100.0 & 100.0 & 100.0 & 100.0 & 100.0 & 100.0 \\
		Weib (eq) & $\delta =1$ & (21$\%$,32$\%$) & 1 & 32.8 & 28.5 & 26.2 & 30.8 & 27.2 & 27.2 & 29.3 & 24.5 & 27.4 \\ 
		&&& 2 & 50.4 & 47.6 & 45.5 & 50.4 & 48.0 & 48.3 & 45.8 & 43.5 & 45.8 \\ 
		&& & 4 & 78.2 & 77.0 & 75.9 & 78.2 & 77.3 & 77.5 & 75.5 & 74.8 & 75.9 \\
		& $\delta =2$ & (21$\%$,52$\%$) & 1 & 87.8 & 84.6 & 82.2 & 87.5 & 84.4 & 84.2 & 84.5 & 80.7 & 85.1 \\ 
		&& & 2 & 99.2 & 99.1 & 98.9 & 99.3 & 99.2 & 99.1 & 98.8 & 98.5 & 98.9 \\ 
		&& & 4 & 100.0 & 100.0 & 100.0 & 100.0 & 100.0 & 100.0 & 100.0 & 100.0 & 100.0 \\ 
		\\
		\multicolumn{13}{c}{\em S5: Weibull (different shape)}\\ [6pt]
		Weib (uneq) & $\delta =1$ & (8$\%$,40$\%$) & 1 & 25.8 & 22.4 & 25.7 & 25.0 & 22.0 & 22.7 & 24.5 & 20.3 & 19.4 \\ 
		&& & 2 & 35.7 & 33.4 & 39.7 & 38.5 & 36.3 & 38.1 & 37.1 & 35.0 & 32.8 \\
		&& & 4 & 55.7 & 54.7 & 62.0 & 60.1 & 59.1 & 61.8 & 62.5 & 60.8 & 57.9 \\
		& $\delta = 2$ & (8$\%$,42$\%$) & 1 & 75.4 & 71.5 & 73.3 & 77.2 & 73.2 & 73.1 & 77.8 & 72.7 & 72.3 \\
		&& & 2 & 94.9 & 94.1 & 95.4 & 96.0 & 95.6 & 95.8 & 96.1 & 95.3 & 95.2 \\
		&& & 4 & 99.9 & 99.8 & 99.9 & 99.9 & 99.9 & 99.9 & 100.0 & 99.9 & 99.9 \\ 
		Unif (eq) & $\delta =1$ & (29$\%$,48$\%$) & 1 & 25.0 & 20.5 & 21.6 & 25.2 & 20.6 & 20.6 & 23.9 & 19.4 & 18.1 \\
		&& & 2 & 35.8 & 33.3 & 35.2 & 37.9 & 35.3 & 34.8 & 38.4 & 35.5 & 33.6 \\ 
		&& & 4 & 59.0 & 57.7 & 60.4 & 60.9 & 59.8 & 59.2 & 61.2 & 59.9 & 56.3 \\ 
		& $\delta =2$ & (29$\%$,50$\%$) & 1 & 75.5 & 70.6 & 70.0 & 77.2 & 72.9 & 72.8 & 74.5 & 67.7 & 70.3 \\
		&& & 2 & 94.3 & 93.5 & 93.7 & 96.3 & 95.6 & 95.8 & 95.6 & 94.6 & 94.7 \\ 
		&& & 4 & 99.9 & 99.9 & 99.9 & 99.9 & 99.9 & 99.9 & 100.0 & 100.0 & 100.0 \\ 
		Weib (eq) & $\delta =1$ & (13$\%$,42$\%$) & 1 & 26.4 & 22.9 & 23.9 & 27.5 & 24.2 & 24.0 & 26.1 & 22.4 & 21.3 \\ 
		&&& 2 & 40.2 & 37.9 & 40.8 & 40.3 & 38.2 & 38.1 & 43.0 & 40.9 & 37.9 \\ 
		&& & 4 & 63.0 & 61.6 & 65.0 & 68.3 & 67.4 & 67.4 & 67.9 & 66.9 & 63.4 \\ 
		& $\delta =2$ & (13$\%$,44$\%$) & 1 & 81.0 & 77.3 & 77.4 & 81.6 & 78.4 & 78.3 & 82.2 & 78.5 & 79.3 \\ 
		&& & 2 & 97.5 & 96.9 & 97.2 & 98.0 & 97.6 & 97.7 & 97.9 & 97.7 & 97.6 \\ 
		&& & 4 & 100.0 & 100.0 & 100.0 & 100.0 & 100.0 & 100.0 & 100.0 & 100.0 & 100.0 \\ 
		\bottomrule
	\end{tabular}
\end{table}

\begin{table}[ht!]
	\small
	\renewcommand{\arraystretch}{0.8}
	\centering
	\caption{Power values in $\%$ (nominal level $\alpha = 5\%$) under the alternative $\mu_2 -\mu_1 = \delta \in\{1,2\}$ for the asymptotic (Asym), our studentized permutation (st P) and the unstudentized permutation (un P) tests in Scenarios S6 and S7. }\label{tab:power_weibull}
	\begin{tabular}{ccccccccccccc}
		\\ \toprule
		&&&&  \multicolumn{3}{c}{$\mathbf{n} = K \cdot(24,16)$} & \multicolumn{3}{c}{$\mathbf{n} = K \cdot(16,24)$} &  \multicolumn{3}{c}{$\mathbf{n} = K \cdot(20,20)$} \\
		\cmidrule(lr){5-7}\cmidrule(lr){8-10}\cmidrule(lr){11-13} Cens. Distr & $\delta$ & Cens. rates & $K$ & Asym & st P & un P & Asym & st P & un P & Asym & st P & un P  \\
		\hline
		\\
		\multicolumn{13}{c}{\em  S6: Weibull (different scale)}\\ [6pt]
		Weib (uneq) & $\delta =1$ & (8$\%$,40$\%$) & 1 & 25.3 & 21.7 & 25.5 & 24.8 & 21.2 & 22.0 & 23.5 & 19.2 & 18.2 \\
		&& & 2 & 35.5 & 33.2 & 39.7 & 38.9 & 36.8 & 38.7 & 38.1 & 35.4 & 33.3 \\ 
		&& & 4 & 57.3 & 55.9 & 63.7 & 61.2 & 60.0 & 62.1 & 62.7 & 60.8 & 57.8 \\
		& $\delta = 2$ & (8$\%$,48$\%$) & 1 & 70.5 & 66.1 & 70.0 & 72.3 & 68.5 & 69.0 & 73.5 & 68.2 & 66.7 \\ 
		&& & 2 & 90.4 & 89.2 & 92.0 & 92.3 & 91.5 & 92.0 & 93.7 & 92.6 & 91.6 \\ 
		&& & 4 & 99.7 & 99.5 & 99.8 & 99.8 & 99.7 & 99.8 & 99.9 & 99.9 & 99.9 \\ 
		Unif (eq) & $\delta =1$ & (29$\%$,48$\%$) & 1 & 24.5 & 20.3 & 21.4 & 24.5 & 20.5 & 20.2 & 23.8 & 19.0 & 18.2 \\ 
		&& & 2 & 36.1 & 33.7 & 35.4 & 37.6 & 35.2 & 34.9 & 37.1 & 34.3 & 31.8 \\
		&& & 4 & 57.9 & 56.5 & 59.4 & 61.2 & 60.0 & 59.7 & 61.4 & 60.1 & 56.7 \\ 
		& $\delta =2$ & (29$\%$,67$\%$) & 1 & 69.6 & 64.5 & 65.2 & 73.0 & 68.3 & 67.8 & 69.4 & 63.6 & 64.1 \\ 
		&& & 2 & 91.6 & 90.7 & 91.0 & 92.8 & 91.9 & 91.8 & 92.5 & 91.5 & 90.8 \\ 
		&& & 4 & 99.6 & 99.5 & 99.7 & 99.9 & 99.9 & 99.8 & 99.8 & 99.8 & 99.8 \\ 
		Weib (eq) & $\delta =1$ & (13$\%$,42$\%$) & 1 & 26.5 & 22.9 & 24.3 & 25.8 & 22.3 & 22.2 & 26.5 & 22.5 & 21.4 \\ 
		&&& 2 & 40.4 & 38.0 & 40.5 & 42.4 & 40.6 & 40.2 & 42.1 & 39.4 & 36.8 \\
		&& & 4 & 63.7 & 62.6 & 65.7 & 67.9 & 66.8 & 66.5 & 68.2 & 67.1 & 63.4 \\ 
		& $\delta =2$ & (13$\%$,65$\%$) & 1 & 74.9 & 71.2 & 72.9 & 77.7 & 74.1 & 73.6 & 75.8 & 71.7 & 71.7 \\ 
		&& & 2 & 94.3 & 93.4 & 94.5 & 96.0 & 95.6 & 95.6 & 95.3 & 94.7 & 94.1 \\
		&& & 4 & 99.8 & 99.8 & 99.8 & 100.0 & 100.0 & 100.0 & 99.9 & 99.9 & 99.8 \\ 
		\\
		\multicolumn{13}{c}{\em S7: Weibull vs. piece-wise Exponential}\\ [6pt]
		Weib (uneq) & $\delta =1$ & (7$\%$,46$\%$) & 1 & 16.9 & 14.3 & 18.0 & 17.4 & 14.5 & 15.6 & 17.5 & 14.6 & 12.4 \\ 
		&& & 2 & 24.2 & 22.3 & 28.9 & 26.3 & 24.6 & 26.0 & 24.9 & 22.5 & 19.1 \\ 
		&& & 4 & 37.0 & 36.1 & 44.4 & 42.3 & 41.0 & 43.5 & 43.4 & 42.0 & 38.0 \\ 
		& $\delta =1$ & (7$\%$,52$\%$) & 1 & 46.1 & 41.7 & 47.3 & 49.8 & 45.4 & 45.9 & 49.3 & 43.4 & 39.8 \\ 
		&& & 2 & 70.0 & 67.4 & 74.3 & 74.0 & 72.0 & 72.9 & 75.1 & 72.9 & 68.6 \\
		&& & 4 & 92.2 & 91.9 & 94.5 & 95.4 & 95.0 & 95.5 & 95.9 & 95.6 & 94.4 \\ 
		Unif (eq) & $\delta =1$ & (25$\%$,59$\%$) & 1 & 17.4 & 14.3 & 16.1 & 17.2 & 14.4 & 13.9 & 17.5 & 14.0 & 11.8 \\ 
		&& & 2 & 25.1 & 23.2 & 25.9 & 26.0 & 24.0 & 23.4 & 25.8 & 23.5 & 19.5 \\ 
		&& & 4 & 38.2 & 37.1 & 40.9 & 40.5 & 39.6 & 38.7 & 44.3 & 43.0 & 37.5 \\ 
		& $\delta =2$ & (25$\%$,69$\%$) & 1 & 48.6 & 43.6 & 46.0 & 49.1 & 44.7 & 43.5 & 48.6 & 42.9 & 40.3 \\ 
		&& & 2 & 71.9 & 69.4 & 72.6 & 73.3 & 71.2 & 70.1 & 76.0 & 73.8 & 69.6 \\ 
		&& & 4 & 93.2 & 92.9 & 94.1 & 95.5 & 95.1 & 95.0 & 96.0 & 95.5 & 94.2 \\ 
		Weib (eq) & $\delta =1$ & (11$\%$,56$\%$) & 1 & 18.0 & 14.9 & 17.8 & 18.9 & 15.9 & 15.8 & 17.8 & 14.5 & 12.3 \\ 
		&&& 2 & 26.3 & 24.7 & 28.5 & 27.5 & 25.1 & 25.1 & 28.0 & 25.8 & 21.9 \\
		&& & 4 & 41.5 & 40.3 & 45.7 & 44.1 & 43.2 & 42.9 & 46.3 & 45.1 & 39.5 \\ 
		& $\delta =2$ & (11$\%$,67$\%$) & 1 & 50.7 & 45.9 & 50.1 & 51.7 & 47.5 & 47.3 & 53.4 & 48.1 & 44.7 \\
		&& & 2 & 71.9 & 70.1 & 74.1 & 78.3 & 76.7 & 76.6 & 79.1 & 76.8 & 73.2 \\
		&& & 4 & 95.0 & 94.8 & 95.8 & 96.1 & 95.8 & 95.9 & 96.7 & 96.5 & 95.6 \\ 
	
		\bottomrule
	\end{tabular}
\end{table}

\begin{figure}
	\centering
	\includegraphics[width=\textwidth]{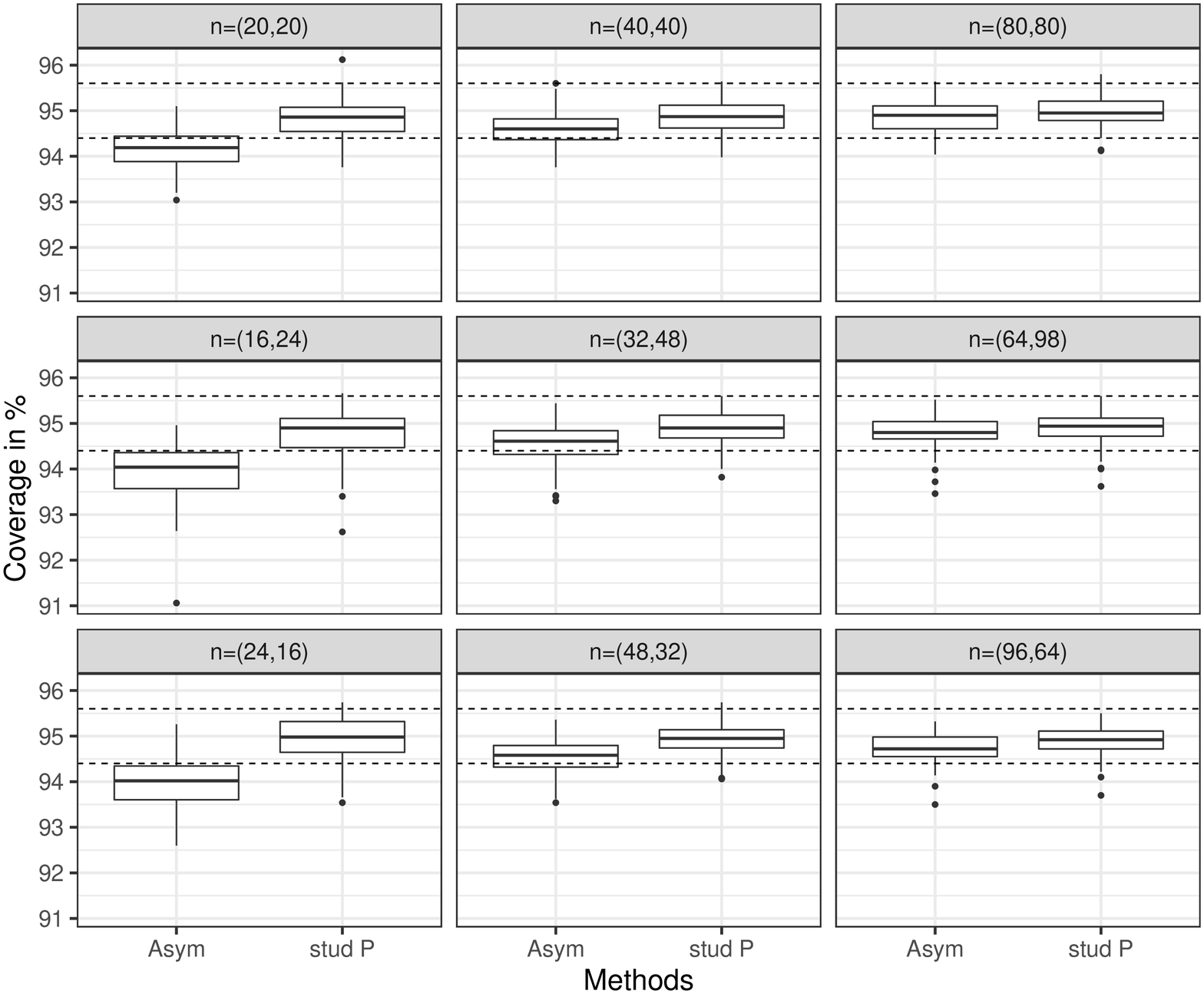}
	\caption{Coverage in $\%$ (nominal level $\alpha=5\%$) of the confidence intervals for the ratio $\mu_1/\mu_2$ based on the asymptotic approximation (Asym) and the  studentized permutation approach (stud P). The dashed, horizontal lines represent the binomial 95$\%$-confidence interval $[94.4\%,95.6\%]$} \label{fig:CI_rat}
\end{figure}

\clearpage

\section{Counting process notation}
For the proofs, we adopt the counting process notation of \cite{abgk}. Let $N_i(t)=\sum_{j=1}^{n_i}\delta_{ij}\mathbf{1}\{X_{ij}\leq t\}$ be the number of observed events up until $t$ in group $i=1,2$ and $Y_i(t)=\sum_{j=1}^{n_i}\mathbf{1}\{X_{ij}\geq t\}$ denotes the number of individuals under risk just before $t$ in group $i=1,2$. Moreover, let $N=N_1+N_2$ and $Y=Y_1+Y_2$ be the respective versions for the pooled sample. It is easy to check that
\begin{align}\label{eqn:Yi=SiGi}
\widehat S_{i-}(t)\widehat G_{i-}(t) = \frac{1}{n_i}\sum_{j=1}^{n_i}\mathbf{1}\{X_{ij}\geq t\}= \frac{1}{n_i} Y_i(t).
\end{align}
Given the counting process notation, we can write the Kaplan--Meier and Nelson--Aalen estimators as follows
\begin{align*}
\widehat S_i(t)=\prod_{k:t_{ik}\leq t}\Bigl( 1- \frac{\Delta N_i(t_{ik})}{Y_i(t_{ik})} \Bigr), \quad \widehat A_{i}(t)=\sum_{k:t_{ik}\leq t} \frac{\Delta N_i(t_{ik})}{Y_i(t_{ik})}\quad (i=1,2;\,t\geq 0),
\end{align*}
where $\Delta N_i(t) = N_i(t) - N_{i-}(t)$ is the increment of $N_i$ in $t$ and $t_{i1},t_{i2},\ldots$ are the distinctive time points within the observed times $(X_{ij})_j$ of group $i$. Moreover, we introduce their pooled counterparts:
\begin{align*}
\widehat S(t)=\prod_{k:t_k\leq t}\Bigl( 1- \frac{\Delta N(t_k)}{Y(t_k)} \Bigr), \quad \widehat A(t)=\sum_{k:{t_k}\leq t} \frac{\Delta N(t_k)}{Y(t_k)}\quad (t\geq 0),
\end{align*}
where $t_1,\ldots,t_d$ are the distinctive time points within the pooled observation times $\mathbf{X}$.

\section{Proof of \eqref{eqn:asym_diff} and \eqref{eqn:asym_ratio}}

The convergence in \eqref{eqn:asym_diff} and \eqref{eqn:asym_ratio} directly follow from the continuous mapping theorem, the $\delta$-method and the following Proposition.
\begin{prop}\label{prop:uncond}
	As $n\rightarrow \infty$,
	$
	\sqrt{n}( \widehat{\mu}_i -  \mu_i) \stackrel d \longrightarrow Z_i \sim \mathcal{N}(0, \sigma_i^2)
	$ 
	with variance
	\begin{align}\label{eqn:def_sigmai}
		\sigma^2_i =  \kappa_i^{-1} \int_0^\tau \Big(\int_x^\tau S_i(t) \,\mathrm{ d }t\Big)^2 \frac{1}{(1-\Delta A_i(x))G_{i-}(x)S_{i-}(x)} \;\mathrm{ d }A_i(x).
	\end{align}
\end{prop}
\begin{proof}[Proof of Proposition \ref{prop:uncond}]
	Let $\mathbb{D}$ be the Skorohod space consisting of all c\`{a}dl\`{a}g functions on $[0,\tau]$.  By Example 3.9.31 of \cite{vaartWellner1996}
	\begin{align}\label{eqn:process_conv_F_uncon}
	\sqrt{n_i}(\widehat S_{i} - S_i) \overset{d}{\to} \mathbb{G}_i \text{ on }\mathbb D
	\end{align}
	for a centered Gaussian process $\mathbb{G}_i$ with covariance structure
	\begin{align*}
	(s,t)\mapsto S_i(t)S_i(s) \int_0^{\min(s,t)} \frac{1}{(1-\Delta A_i(x))G_{i-}(x)S_{i-}(x)}\mathrm{ d }A_i(x).
	\end{align*}
	Thus, we can deduce from \eqref{eqn:process_conv_F_uncon} and the continuous mapping theorem
	\begin{align*}
	\sqrt{n}( \widehat{\mu}_i -  \mu_i ) =  \sqrt{\frac{n}{n_i}}   \int_0^\tau \sqrt{n_i}(\widehat S_i(t) - S_i(t)) \,\mathrm{ d }t  \overset{ d}{\longrightarrow}   \kappa_i^{-1/2}\int_0^\tau \mathbb{G}_i(t)\,\mathrm{ d }t =Z_i.
	\end{align*}
	By Fubini's Theorem \citep[Sec. 3.9.2]{vaartWellner1996}, $Z_i$ is indeed centered normally distributed with variance given by
	\begin{align*}
	\sigma_i^2 &= \kappa_i^{-1}\int_0^{\tau}\int_0^\tau \E( \mathbb{G}_i(t)\mathbb{G}_i(s) ) \,\mathrm{ d }t\,\mathrm{ d }s \\
	& = \kappa_i^{-1}\int_0^{\tau}\int_0^\tau S_i(t)S_i(s) \int_0^{\min(s,t)} \frac{1}{(1-\Delta A_i(x))G_{i-}(x)S_{i-}(x)}\,\mathrm{ d }A_i(x) \,\mathrm{ d }t \,\mathrm{ d }s \\
	&=  \kappa_i^{-1}\int_0^\tau \Bigl(\int_x^\tau S_i(t) \,\mathrm{ d }t \Bigr)^2 \frac{1}{(1-\Delta A_i(x))G_{i-}(x)S_{i-}(x)}\,\mathrm{ d }A_i(x).
	\end{align*}
\end{proof}

\section{Proof of the variance estimator's consistency}

Define $y_i=S_{i-}G_{i-}$ and $\nu_i$ by $\nu_i(t) = \int_0^t G_{i-}(s) \,\mathrm{ d }F_i(s)$ $(t\geq 0)$. By the Glivenko--Cantelli Theorem
\begin{align}\label{eqn:Yi+Ni}
\sup_{t\in[0,\tau]}|n_i^{-1}Y_i(t) - y_i(t)| + 	\sup_{t\in[0,\tau]}|n_i^{-1}N_i(t) - \nu_i(t)| \to 0 \quad \text{as }n\to\infty
\end{align}
almost surely. It is well known that this combined with the continuous mapping theorem implies the uniform consistency of the Kaplan--Meier and Nelson--Aalen estimators:
\begin{align}\label{eqn:conv_Si}
\sup_{t\in[0,\tau]}|\widehat S_i(t) - S_i(t)| + \sup_{t\in[0,\tau]}|\widehat A_i(t) - A_i(t)| \to 0\quad \text{as }n\to\infty
\end{align}
almost surely. Obviously, it follows that
\begin{align}\label{eqn:intS_conv}
\sup_{x\in[0,\tau]}\Bigl|\int_x^\tau \widehat S_i(t) \,\mathrm{ d }t -\int_x^\tau S_i(t) \,\mathrm{ d }t\Bigr| \to 0\quad \text{as }n\to\infty
\end{align}
almost surely. In particular,
\begin{align}\label{eqn:var_cons_mui}
\widehat \mu_i = \int_0^\tau \widehat S_i(t) \,\mathrm{ d }t \to \int_0^\tau S_i(t) \,\mathrm{ d }t = \mu_i\quad \textrm{ as }n\to\infty
\end{align}
with probability one. Moreover, we can deduce from \eqref{eqn:Yi=SiGi}, \eqref{eqn:Yi+Ni} and \eqref{eqn:intS_conv} that we have almost surely as $n\to \infty$
\begin{align}
\widehat \sigma_i^2 &= \frac{n}{n_i} \int_0^\tau \Big(\int_x^\tau \widehat S_i(t) \,\mathrm{ d }t\Big)^2 \frac{1}{(1-\Delta \widehat A_i(x))n_i^{-1}Y_i(x)} \;\mathrm{ d }\widehat A_i(x) \nonumber \\
&\to \kappa_i^{-1} \int_0^\tau \Big(\int_x^\tau S_i(t) \,\mathrm{ d }t\Big)^2 \frac{1}{(1-\Delta A_i(x))S_{i-}(x)G_{i-}(x)} \;\mathrm{ d } A_i(x) = \sigma^2_i. \label{eqn:var_cons_sigmai}
\end{align}
Clearly, the consistency of $\widehat \sigma^2=\widehat \sigma_1^2+\widehat \sigma_2^2$ follows. In the same way, we can deduce the consistency of $\widehat \sigma^2_{\text{rat}}$ which was defined after Equation \eqref{eqn:asym_ratio}.

\section{Proof of Theorems~\ref{theo:perm_unstud} and \ref{theo:perm}}
We first introduce the limits of $Y/n, N/n$, $\widehat S$ and $\widehat A$:
\begin{align*}
&y(t)=\kappa_1y_1(t)+\kappa_2y_2(t), \quad \nu(t)= \kappa_1\nu_1(t)+ \kappa_2\nu_2(t), \\
&S(t) = \exp\Bigl\{ - \int_0^t \frac{1}{y(s)}\mathrm{d}\nu(s) \Bigr\}, \quad A(t) = \int_0^t \frac{1}{y(s)}\mathrm{d}\nu(s),
\end{align*}
where $y_i(t)=S_{i-}(t)G_{i-}(t)$ and $\nu_i(t) = \int_0^t G_{i-}(s) \,\mathrm{ d }F_i(s)$ were already defined in the proof of Proposition~\ref{prop:uncond}. In fact, from the Glivenko-Cantelli Theorem (and the continuous mapping theorem for the convergence of $\widehat S$) we obtain immediately
\begin{align}\label{eqn:conv_Y+N+S}
\sup_{t\in[0,\tau]} \Bigl | \frac{N(t)}{n} - \nu(t) \Bigr| + \sup_{t\in[0,\tau]} \Bigl | \frac{Y(t)}{n} - y(t) \Bigr| + \sup_{t\in[0,\tau]} \Bigl | \widehat S(t) - S(t) \Bigr| + + \sup_{t\in[0,\tau]} \Bigl | \widehat A(t) - A(t) \Bigr|  \to 0
\end{align}
with probability one as $n\to\infty$. In particular,
\begin{align*}
\widehat \mu = \int_0^\tau \widehat S(t) \,\mathrm{ d } t \to \int_0^\tau S(t) \,\mathrm{ d }t = \mu
\end{align*}
almost surely as $n\to\infty$.

For the first step of the proof, we follow the argumentation of the previous proof of \eqref{eqn:asym_diff}. As explained by \cite{dopa2018} (see Theorem 5 in their supplement), the following conditional convergence is a straightforward consequence of Theorems 3.7.1 and 3.7.2 in \cite{vaartWellner1996}:
\begin{align*} 
&\sqrt{n} \Bigl( \widehat S^\pi_{1} - \widehat S,\widehat S^\pi_{2} - \widehat S \Bigr) \overset{ d}{\longrightarrow} \mathbb{G}^\pi\quad \text{on }\mathbb D^2 \textrm{ as }n\to\infty
\end{align*}
given the data in probability, where $\mathbb{G}^\pi=(\mathbb{G}^\pi_1,\mathbb{G}^\pi_2)$ is a centered Gaussian process on $\mathbb D^2$ with covariance structure given by
\begin{align*}
\E(\mathbb{G}^\pi_i(s) \mathbb{G}^\pi_{i'}(t)) = \Big(\frac{1}{\kappa_i}\mbf{1}\{i=i'\} - 1\Big)  S(t)S(s) \int_{0}^{\min(s,t)} \frac{1}{(1-\Delta A(x))y(x)}\mathrm{ d } A(x).
\end{align*}
Consequently, we obtain from the continuous mapping theorem that given the data in probability
\begin{align}\label{eqn:perm_Z}
\sqrt{n}(\widehat\mu_1^\pi - \widehat \mu, \widehat \mu_2^\pi - \widehat \mu) \overset{ d}{\rightarrow} \Bigl(\int_0^\tau \mathbb{G}^\pi_1(s)\,\mathrm{ d }s, \int_0^\tau \mathbb{G}^\pi_2(s) \,\mathrm{ d }s \Bigr) = (Z_1^\pi,Z_2^\pi),
\end{align}
where $(Z_1^\pi,Z_2^\pi)$ is $2-$dimensional, centered normally distributed with covariance structure
\begin{align*}
E(Z_i^\pi Z_{i'}^\pi) = \Big(\frac{1}{\kappa_i}\mbf{1}\{i=i'\} - 1\Big)\sigma^{\pi 2},\quad \sigma^{\pi 2} = \int_0^\tau \Big(\int_x^\tau S(t) \,\mathrm{ d }t\Big)^2 \frac{1}{(1-\Delta A(x))y(x)} \;\mathrm{ d } A(x).
\end{align*}
Applying again the continuous mapping theorem yields that given the data in probability
\begin{align}\label{eqn:perm_diff}
\sqrt{n} (\widehat \mu_1^\pi - \widehat \mu_2^\pi) \overset{ d}{\rightarrow} Z_1^\pi - Z_2^\pi \sim N(0,\sigma^{\pi 2}_{\text{diff}}) \text{ with } \sigma^{2}_{\text{perm}} = \frac{\sigma^{\pi 2}}{\kappa_1\kappa_2} \textrm{ as }n\to\infty.
\end{align}
This proves Theorem~\ref{theo:perm_unstud}.

To verify Theorem~\ref{theo:perm}, it remains to discuss the consistency of the variance estimator. Therefor, we fix the original observations $({\bm X},{\bm \delta})$. Note that $N$, $Y$ and $S$ does not change when permuting the data. Thus, we can treat them all as fixed functions. Moreover, we can assume without a loss of generality that \eqref{eqn:conv_Y+N+S} holds for them.
Following \cite[equation 6.1]{neuhaus:1993}, we can deduce
\begin{align*}
\sup_{t\in[0,\tau]} \Bigl | \frac{Y_i^\pi(t)}{Y(t)} - \kappa_i \Bigr| \overset{p}{\rightarrow} 0 \textrm{ as }n\to\infty.
\end{align*}
Using similar arguments, the statement remains true for $N_i^\pi/N$. Combining both, \eqref{eqn:conv_Y+N+S} and the continuous mapping theorem yields
\begin{align*}
\sup_{t\in[0,\tau]} \Bigl | \frac{Y_i^\pi(t)}{n} - \kappa_iy(t) \Bigr| + \sup_{t\in[0,\tau]} \Bigl | \frac{N_i^\pi(t)}{n} - \kappa_i\nu(t) \Bigr| + \sup_{t\in[0,\tau]} \Bigl | \widehat S_i^\pi(t) - S(t) \Bigr| + \sup_{t\in[0,\tau]} \Bigl | \widehat A_i^\pi(t) - A(t) \Bigr|  \overset{p}{\rightarrow} 0.
\end{align*}
In particular, we obtain
\begin{align*}
\Bigl | \widehat \mu_i^\pi - \mu \Bigr| + \sup_{t\in[0,\tau]} \Bigl | \int_t^\tau\widehat S_i^\pi(s)\,\mathrm{ d }s - \int_t^\tau S(s)\,\mathrm{ d }s \Bigr| \overset{p}{\rightarrow} 0.
\end{align*}
Combining all previous statements yields that as $n\to\infty$
\begin{align*}
\widehat \sigma_i^{\pi 2} &= \frac{n}{n_i} \int_0^\tau \Big(\int_x^\tau \widehat S_i^\pi(t) \,\mathrm{ d }t\Big)^2 \frac{1}{n_i^{-1}Y_i^{\pi}(x)} \;\mathrm{ d }\widehat A_i^\pi(x) \nonumber \\
&\overset{p}{\rightarrow} \kappa_i^{-1}\int_0^\tau \Big(\int_x^\tau S(t) \,\mathrm{ d }t\Big)^2 \frac{1}{y(x)} \;\mathrm{ d } A(x) = \kappa_i^{-1} \sigma^{\pi 2}.
\end{align*}
Finally, the desired convergence of the variance estimator follows, i.e. as $n\to\infty$
\begin{align*}
&\widehat \sigma^{\pi 2} = \widehat\sigma_1^{\pi 2} + \widehat\sigma_2^{\pi 2} \overset{p}{\rightarrow}  \kappa_1^{-1} \sigma^{\pi 2} + \kappa_2^{-1} \sigma^{\pi 2} = \sigma^{2}_{\text{perm}}.
\end{align*}

\section{Proof of Corollary~\ref{cor:tests+confidence_intervals;rat}}
It is sufficient to show that given the data in probability
\begin{align}
\sqrt{n} (1/ \widehat\sigma_{\text{rat}}^\pi)[\log(\widehat \mu_1^\pi) - \log(\widehat \mu_2^\pi)] \overset{ d}{\rightarrow} Z_{\text{rat}}^\pi \sim N(0,1) \textrm{ as }n\to\infty.
\end{align}
The corresponding proof can again be separated into two steps: (1) verification of the asymptotic normality of $\log(\widehat \mu_1^\pi) - \log(\widehat \mu_2^\pi)$ and (2) showing the consistency of the variance estimator $\widehat \sigma_{\text{rat}}^{\pi 2}$. It is easy to see that (1) follows immediately from \eqref{eqn:perm_Z} and the $\delta$-method. Moreover, (2) can be proven in the same way as in the previous proof and, thus, it is omitted here.

\end{document}